\documentclass[aps,prl,secnumarabic,superscriptaddress,twocolumn,floatfix,10pt]{revtex4-2}

\usepackage{amsmath,amsthm,amsfonts,amssymb}
\usepackage{mathtools}
\usepackage{bbold}
\usepackage{braket}
\usepackage{cancel}

\usepackage{xcolor}   % prefer xcolor over color
\usepackage{graphicx}
\usepackage[inkscape=off]{svg}

\usepackage{soul}
\usepackage{xurl}

\usepackage{silence}
\usepackage{nameref}

\usepackage[pdftex,pagebackref=false,breaklinks=true]{hyperref} 
\hypersetup{
    plainpages=false,       % needed if Roman numbers in frontpages
    unicode=false,          % non-Latin characters in Acrobat’s bookmarks
    pdftoolbar=true,        % show Acrobat’s toolbar?
    pdfmenubar=true,        % show Acrobat’s menu?
    pdffitwindow=false,     % window fit to page when opened
    pdfstartview={FitH},    % fits the width of the page to the window
    pdftitle={PT-like phase transitions from square roots of supersymmetric Hamiltonians},   
    pdfauthor={Jacob L. Barnett and Ramy El-Ganainy},  
    bookmarksnumbered = true, 
    linktoc = none,
    pdfnewwindow=true,      % links in new window
    colorlinks=true,        % false: boxed links; true: colored links
    linkcolor=blue,         % color of internal links
    citecolor=green,        % color of links to bibliography
    filecolor=magenta,      % color of file links
    urlcolor=cyan           % color of external links
}

\usepackage[capitalize]{cleveref}

\newtheorem{theorem}{Theorem}
\newtheorem{lemma}[theorem]{Lemma}

\begin{document}

\title{
\texorpdfstring{$\mathcal{PT}$-like phase transitions from square roots of supersymmetric Hamiltonians}{PT-like phase transitions from square roots of supersymmetric Hamiltonians}
}

\author{Jacob L. Barnett }
\affiliation{Basque Center for Applied Mathematics, Bilbao, Bizkaia, 48009 Spain}
\email{jbarnett@bcamath.org}

\author{Ramy El-Ganainy}
\affiliation{Department of Electrical and Computer Engineering, Saint Louis University,  Saint Louis, MO 63103, USA}
\email{relganainy@slu.edu}

\begin{abstract}
We introduce a general framework for realizing $\mathcal{PT}$-like phase transitions in non-Hermitian systems without imposing explicit parity--time ($\mathcal{PT}$) symmetry. The approach is based on constructing a Hamiltonian as the square root of a supersymmetric partner energy-shifted by a constant. This formulation naturally leads to bipartite dynamics with balanced gain and loss and can incorporate non-reciprocal couplings. The resulting systems exhibit entirely real spectra over a finite parameter range precisely when the corresponding passive Hamiltonian lacks a zero mode. As the non-Hermitian parameter representing gain and loss increases, the spectrum undergoes controlled real-to-complex transitions at second-order exceptional points. We demonstrate the versatility of this framework through several examples, including well-known models such as the Hatano--Nelson (HN) and complex Su--Schrieffer--Heeger (cSSH) lattices. Extending the formalism to $q$-commuting matrices further enables the systematic realization of higher-order exceptional points in systems with unidirectional couplings. Overall, this work uncovers new links between non-Hermitian physics and supersymmetry, offering a practical route to engineer photonic arrays with tunable spectral properties beyond what is achievable with explicit $\mathcal{PT}$-symmetry.
\end{abstract}

%We present a framework for realizing $\mathcal{PT}$-like phase transitions in non-Hermitian systems without imposing explicit parity–time symmetry. Our construction is based on taking a square root of a supersymmetric Hamiltonian shifted by a constant, which naturally yields bipartite dynamics with balanced gain and loss and possibly non-reciprocal couplings. These systems exhibit real spectra over a finite parameter range if and only if the passive Hamiltonian does not possess a zero mode. Controlled real-to-complex spectral transitions occur when the gain parameter crosses second-order exceptional points. We show that familiar models, such as the Hatano–Nelson and Su-Schrieffer-Heeger models, arise as special cases.  Extending this idea, we demonstrate that $q$-commuting matrices allow the systematic generation of higher-order exceptional points in systems with one-directional couplings. Our results provide a practical route to engineer tunable spectral phase transitions in standard photonic platforms, establishing new connections between non-Hermitian physics, supersymmetry, and optical implementations.
    
\maketitle

\section{Introduction}

%{\color{red} J.B. I've been using the term "passive" to refer to $\gamma = 0$, is this correct?} \textcolor{blue}{Using passive for $\gamma=0$ is OK. Though we have to comment on that for the HN array because it is not passive even when $\gamma=0$.

Non-Hermitian physics has attracted significant attention in recent years as a powerful framework for modeling open systems—systems that exchange energy or information with their environment. Among the most intriguing subclasses of non-Hermitian systems are those exhibiting \emph{parity-time} ($\mathcal{PT}$) symmetry~\cite{Bender1998,Bender1999}. A system is said to be $\mathcal{PT}$-symmetric if its governing Hamiltonian is invariant under the combined operations of \emph{parity inversion} ($\mathcal{P}$) and \emph{time reversal} ($\mathcal{T}$), i.e., $H \mathcal{PT} = \mathcal{PT} H$, where $\mathcal{P}$ and $\mathcal{T}$ are isometric and involutory operators representing parity and time-reversal transformations, respectively. The operator $\mathcal{P}$ is linear, while $\mathcal{T}$ is antilinear and  acts as complex-conjugation in a suitable basis \cite{Garcia2005}. Remarkably, a $\mathcal{PT}$-symmetric Hamiltonian can exhibit an entirely real eigenvalue spectrum despite being non-Hermitian. This occurs in the so-called $\mathcal{PT}$-unbroken regime, whose defining characteristic is that the Hamiltonian's eigenspaces are invariant subspaces of the combined $\mathcal{PT}$ operator. As a system's parameters vary,  $\mathcal{PT}$-symmetry may spontaneously break, resulting in the existence of non-real eigenvalues appearing in complex-conjugate pairs whose corresponding eigenspaces are no longer $\mathcal{PT}$-invariant subspaces. The transition point between the $\mathcal{PT}$-unbroken and $\mathcal{PT}$-broken phases occurs at exceptional points (EPs), where multiple eigenvalues and eigenspaces of the Hamiltonian coalesce \cite{Kato1995,Heiss2000,Rotter2003}.

While initial interest in $\mathcal{PT}$-symmetry was primarily theoretical, it has since found widespread applicability across various physical platforms, most notably in optics and photonics \cite{ElGanainy2007,Makris2008,Rter2010}. In particular, optics and photonics have emerged as a fertile ground for the experimental realization of $\mathcal{PT}$-symmetric systems, offering unprecedented control over light beyond the conventional limits of refractive index engineering \cite{Feng2017,ElGanainy2018,zdemir2019,Miri2019}. A key reason for this success lies in the constraints imposed by  $\mathcal{PT}$-symmetry on the complex refractive index: its real part must be an even function of position, while its imaginary part—representing gain and loss—must be odd. In optics, these conditions can be readily satisfied using coupled waveguides or resonators with spatially balanced regions of gain and loss. Moreover, by employing coupled-mode theory, the governing Maxwell equations can be reduced to effective finite-dimensional non-Hermitian Hamiltonians, rendering both design and analysis tractable \cite{Fan2003,Haus1991}.\\

However, $\mathcal{PT}$-symmetry is not the only route to generating non-Hermitian Hamiltonians with real spectra. For instance, any Hamiltonian $H$ that commutes with an arbitrary anti-unitary operator $\Theta$, which can be interpreted as a generalized time-reversal symmetry \cite{Wigner1932Zeitumkehr,Wigner1959,Wigner1993}, shares the same spectral constraints: real eigenvalues correspond to $\Theta$-invariant eigenspaces, while complex eigenvalues occur in conjugate pairs, $(\lambda, \lambda^*)$, with $\Theta$ mapping eigenstates corresponding to $\lambda$ to  eigenstates corresponding to $\lambda^*$. A Hamiltonian respects a generalized time-reversal symmetry if and only if it is \textit{pseudo-Hermitian} \cite{Radjavi1969}, where an operator is called pseudo-Hermitian if there exists a Hermitian operator, $\eta$, such that \cite{Heisenberg1957}
\begin{align}
    H= \eta^{-1} H^{\dagger} \eta. \label{eqn:pseudo-defn}
\end{align} 
In this case, $\eta$, referred to as a metric operator \cite{Scholtz1992}, defines a conserved quantity via its induced quadratic form \cite{Pauli1943DiracFieldQuantization}. The interplay between pseudo-Hermiticity and $\mathcal{PT}$-symmetry has been discussed more recently in \cite{Mostafazadeh2002,Mostafazadeh2003}. %If a $\mathcal{PT}$-symmetric Hamiltonian is transpose-symmetric in some $\mathcal{T}$-real basis, $\eta$ coincides with $\mathcal{P}$. 
A necessary and sufficient condition for a complex matrix to be diagonalizable with real eigenvalues is that a positive-definite choice for $\eta$ exists \cite{Taussky1960,Drazin1962}, in which case $H$ is called \textit{quasi-Hermitian} \cite{dieudonne}. 
As the parameters of a Hamiltonian vary, its associated metric operators can change. Consequently, a Hamiltonian may be quasi-Hermitian only within a restricted region of parameter space. At the boundary of this region, some or all eigenvalues may bifurcate into complex-conjugate pairs.

Although pseudo-Hermiticity is a broader symmetry constraint than \(\mathcal{PT}\)-symmetry, this increased generality has yet to spur widespread experimental implementation. The primary challenge stems from the complexity of realizing generic pseudo-Hermitian Hamiltonians, which typically demand intricate designs featuring asymmetric or long-range couplings as well as precise fine-tuning of numerous parameters. Moreover, additional tuning is often necessary to observe phase transitions. In contrast, \(\mathcal{PT}\)-symmetric systems are more straightforward to engineer, benefiting from a clear geometric interpretation of the combined parity-time operation.

%Although the class of non-Hermitian Hamiltonians with real spectra is broader than the $\mathcal{PT}$-symmetric class, this has not translated into significant experimental efforts. This is primarily due to the complex structure of these Hamiltonians, which poses significant implementation challenges. In the context of discrete optics—i.e., optical systems composed of coupled waveguides or resonators described by coupled mode theory, analogous to tight-binding models—such generalized Hamiltonians often require intricate designs involving long-range and asymmetric couplings. These structural complexities make practical realization extremely challenging, if not impossible. Additionally, observing the phase transition in these systems may require simultaneous tuning of multiple parameters along very specific trajectories, adding to the complexity of practical implementation. As a result, most experimental studies have focused on the narrower but more accessible class of $\mathcal{PT}$-symmetric Hamiltonians. 

In this paper, we introduce an explicit class of pseudo-Hermitian Hamiltonians that exhibit phase transitions from real to complex spectra while remaining experimentally accessible using standard discrete photonic platforms, without the need for asymmetric or long-range couplings. The key idea of our initial construction is to take a nontrivial “square root” of a supersymmetric Hamiltonian energy-shifted by a constant value. The resultant Hamiltonian consists of a sum of anti-commuting Hermitian and anti-Hermitian parts and exhibits a second-order exceptional point across which the spectrum undergoes a transition from real to complex eigenvalues. More generally, we show that sums of $q$-commuting matrices \cite{Koornwinder1997} exhibit $n$-th order exceptional points. 
Beyond its practical implications, our approach establishes a novel connection between non-Hermitian physics, supersymmetric quantum mechanics \cite{Witten1981,Witten1982}, and supersymmetric optics \cite{Miri2013, Heinrich2014,Datta2024}. This connection extends previous research on generating nontrivial Hermitian topologies via square roots \cite{Arkinstall2017,Zhang2019,Kremer2020}. 
\\

%%%%%%%%%%%%%%%%%%%%
\begin{figure}[t]
    \centering
    \includegraphics[width=\columnwidth]{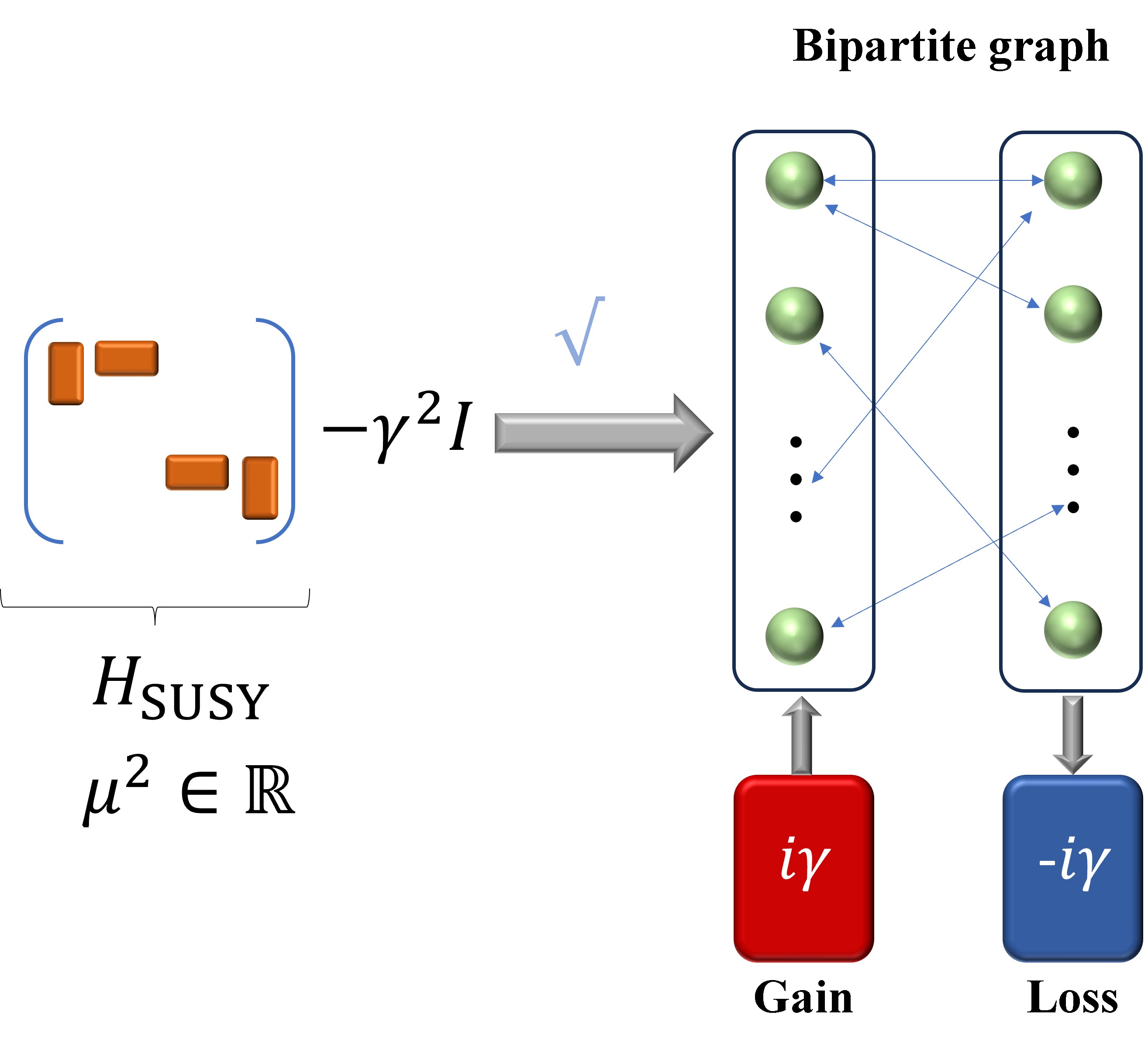}
    \caption{
    Schematic illustration of our construction of non-Hermitian bipartite tight-binding models. Starting with a supersymmetric (SUSY) Hamiltonian, $H_{\mathrm{SUSY}}$, with real spectrum, we take a matrix square root after a uniform downward shift by $\gamma^2$, where $\gamma \in \mathbb{R}$. The result is a tight-binding model with hopping on a bipartite graph and gain and loss applied to vertices in the color classes $+1$ and $-1$, respectively. %The spectrum remains real as long as $\gamma^2$ is less than the smallest eigenvalue of $H_{\mathrm{SUSY}}$. Pairs of eigenvalues bifurcate from the real axis to the imaginary axis via exceptional points as $\gamma$ increases above the square root of an eigenvalue of $H_{\mathrm{SUSY}}$.  
    Our procedure does not require geometric symmetry or explicit $\mathcal{PT}$-symmetry, thereby introducing a more general condition under which real-to-complex spectral transitions occur in non-Hermitian systems.
    }
    \label{fig:Schematic}
\end{figure}
%%%%%%%%%%%%%%%%%%%%

%In optics, one distinguishes between “continuous” and “discrete” models. Continuous systems are governed by differential operators derived from Maxwell's equations, whereas discrete systems arise via coupled mode theory and are represented by finite-dimensional matrices. In this work, we focus on discrete models, which are most relevant to photonic experiments based on coupled waveguides or cavities. See Appendix~\ref{app:continuous} for further discussion of this distinction.

\section{Results}
In what follows, we present the main results of our work, beginning with a description of the general framework for constructing non-Hermitian Hamiltonians with $\mathcal{PT}$-like phase transitions from shifted supersymmetric Hamiltonians. We then analyze the symmetry properties of these systems and illustrate our findings through several representative examples. Finally, we discuss possible generalizations achieved via $q$-deformation.\\

%\noindent
%\textbf{General framework:---}

\subsection{Preliminaries on Bipartite Graphs}

To motivate the discussion, we start by considering the case of $2 \times 2$ pseudo-Hermitian matrices, which can be studied using a set of concrete, mathematically equivalent, and easily verifiable conditions \cite{Mostafazadeh2006PseudoHermitian,Wang2013}. One particularly elegant result is stated below and proven in \cref{app:2x2}.
\begin{theorem} \label{thm:2x2}
    A traceless $2 \times 2$ complex matrix is pseudo-Hermitian if and only if its Hermitian and anti-Hermitian parts anti-commute. 
\end{theorem}
Inspired by Olga Taussky’s remark that the sum of interesting matrices is always worth studying~\cite{Taussky1988}, this result motivates the broader consideration of sums of anti-commuting Hermitian and anti-Hermitian components as a strategy for constructing higher-dimensional pseudo-Hermitian matrices. 

Such a pair of anti-commuting matrices naturally emerges in the context of tight-binding models on bipartite graphs \cite{Coulson1940,Ruedenberg1953}. Before we display this pair in \cref{defn:T-and-chi}, we summarize some requisite elementary definitions from graph theory. A directed graph, $G = (V, E)$, is a set of vertices, $V$, and directed edges, $E \subseteq V \times V$. A bipartite graph is one whose vertices can be partitioned into two disjoint subsets, $V = V_+ \cup V_-$, with $V_+ \cap V_- = \emptyset$, in such a way that no edge connects two vertices in the same subset. Alternatively, bipartite graphs are those that can be assigned a proper 2-coloring, $c:V \to \{\pm 1\}$, in which case we identify $V_\pm = c^{-1}[\{\pm 1\}]$. Prominent examples of bipartite graphs include trees, even cycles, hypercubes, and subgraphs of the hexagonal lattice. To simplify the following discussion, we assume $V$ has finite cardinality; a discussion of the case where $V$ has countably-infinite cardinality is postponed until \cref{app:generalization}.\\

We now introduce the aforementioned pair of anti-commuting matrices, $(T, \chi)$, as block matrices acting on the Hilbert space $\mathbb{C}^{V_+} \oplus \mathbb{C}^{V_-}$,
\begin{align}
    T := \begin{pmatrix}
    0 & T_+ \\
    T_- & 0
\end{pmatrix} &\quad& \chi := \begin{pmatrix}
    \mathbb{1}_{V_+} & 0 \\
    0 & -\mathbb{1}_{V_-}
\end{pmatrix}, \label{defn:T-and-chi}
\end{align}
where $T_{\pm}:\mathbb{C}^{V_\mp} \to \mathbb{C}^{V_\pm}$ are generic linear maps and $\mathbb{1}_{V_\pm}$ denotes the identity operator on $\mathbb{C}^{V_\pm}$. $T$ can be interpreted as a hopping Hamiltonian on the graph $G$ if $T$ is a \textit{weighted adjacency matrix} of $G$, which means its nonzero matrix elements correspond to edges in $G$ or, more explicitly, $(T_{\pm})_{uv} \neq 0 \, \Rightarrow  (v, u) \in E$. Every instance of $T$ is a weighted adjacency matrix of the complete bipartite graph $K_{|V_+|, |V_-|}$. 

Dynamics governed by the Hamiltonian $T$ exhibit \textit{chiral symmetry}, since there exists a unitary operator, $\chi$, that anti-commutes with the Hamiltonian \cite{Kawabata2019}. The chiral operator defines \textit{left-handed} and \textit{right-handed} states (not to be confused with left-handed or right-handed circularly polarized light), which in this case are physically interpreted as the states that are localized to vertices of colors $+1$ and $-1$, respectively. Explicitly, we have
\begin{align}
    \mathcal{H}_+ := \mathbb{C}^{V_+} \oplus \{0_-\}, &\quad& \mathcal{H}_- := \{0_+\} \oplus \mathbb{C}^{V_-},
\end{align}  
where $0_\pm \in \mathbb{C}^{V_{\pm}}$ denotes the zero vector. Equivalently, we could define $ \mathcal{H_{\pm}} = \text{Ran}(1 \pm \chi)$, where $\text{Ran}$ denotes the range of an operator. 

The square of $T$,
\begin{equation}
H_{\mathrm{SUSY}} := T^2 = \begin{pmatrix} T_+ T_- & 0 \\ 0 & T_- T_+ \end{pmatrix},
\end{equation}
is supersymmetric in the sense defined by \cite{Znojil2002,Znojil2004,Bagarello2020}.
When $T_{+} = T_{-}^\dagger$, $H_{\text{SUSY}}$ is supersymmetric in the sense defined by Witten~\cite{Witten1981,Witten1982}. The two blocks $T_+ T_-$ and $T_- T_+$ are isospectral up to zero modes \cite{Sylvester1883}, so the nonzero eigenvalues of $H_{\mathrm{SUSY}}$ are degenerate. This degeneracy is due to the chiral symmetry of $T$; $\chi$ defines a one-to-one correspondence between the eigenspaces of $T$ with eigenvalues $\pm \mu$, which means the subspace generated by these two eigenspaces is an eigenspace of $H_{\mathrm{SUSY}}$ with the eigenvalue $\mu^2$ that is degenerate when $\mu \neq 0$. 

%Lieb's theorem (?) says that the number of zero modes in $T$ is bounded below by the color class imbalance, $||V_+| - |V_-||$ [XXX] \textcolor{red}{Why is Lieb theorem needed here?}. 
%\textcolor{blue}{In the above discussion, notations such as $:=$ and $\simeq$ are not common knowledge in optics community. So the question is: do we need all this formal definition here? I think this section ca be written without all this formality.} \textcolor{red}{J.B. I removed $\simeq$ but left in :=, which just means "this equation is a definition"}

\subsection{Non-Hermitian Hamiltonian and Spectrum}

We now introduce a non-Hermitian perturbation of $T$, 
\begin{align}
H_{\text{NH}} = T + i \gamma \chi &\quad& \gamma \in \mathbb{R}, \label{eqn:HNH-def}
\end{align}
induced by balanced gain and loss on the $+$ and $-$ sites, respectively. This generalizes the class of Hamiltonians studied in \cite{MartnezMartnez2024}, which assumed  $|V_+| = |V_-|$ and that $T_+ = T_-^\dag$ is a matrix with real coefficients. Henceforth, $T$ is referred to as the passive term of $H_{\text{NH}}$. This Hamiltonian 
exemplifies an instance of the "Freshman's dream", \begin{align}
    H_{\text{NH}}^2 = (T + i \gamma \chi)^2 = T^2 - \gamma^2 \chi^2.
\end{align}
Since $\chi^2 = \mathbb{1}$, the spectral properties of $H_{\text{NH}}$ are determined by those of $H_{\text{SUSY}}$. In particular, we use this observation to identify real-to-complex spectral phase transitions in $H_{\mathrm{NH}}$. Given a real eigenvalue, $\mu$, of $H_{\text{SUSY}}$, we find:
\begin{itemize} 
\item If $\mu > \gamma^2$, then $\pm\sqrt{\mu - \gamma^2}$ are real eigenvalues of $H_{\mathrm{NH}}$. 
\item If $0 \neq \mu < \gamma^2$, then the spectrum of $H_{\text{NH}}$ includes purely imaginary pairs, $\pm i \sqrt{\gamma^2 - \mu}$. 
\item If $0 \neq \mu = \gamma^2$, then $H_{\text{NH}}$ exhibits a second-order exceptional point (EP), owing to the existence of a Jordan chain containing a zero mode. This is elaborated upon in \cref{thm:multiplicity}.
\item If $\mu=0$, the kernel of $T$ is generated by two subspaces comprised of left-handed and right-handed zero modes. Each of these subspaces is an eigenspace of $H_{\mathrm{NH}}$ with a purely imaginary eigenvalue: $+i \gamma$ in the left-handed case and $-i \gamma$ in the right-handed case.
\end{itemize}

The zero modes in $T$ are not robust to the non-Hermitian perturbation $i \gamma \chi$. This is in contrast to perturbations respecting chiral symmetry, where the zero modes are symmetry protected \cite{Hasan2010}.

When $H_{\text{SUSY}}$ is quasi-Hermitian, the elements of the spectrum of $H_{\text{NH}}$ appear in complex-conjugate pairs if and only if the numbers of left-handed and right-handed zero modes coincide, i.e. $\dim \ker T_+ = \dim \ker T_-$. If $T$ contains no zero modes and $H_{\text{SUSY}}$ has only real eigenvalues, then a parameter regime for $\gamma$ exists where $H_{\text{NH}}$ has only real eigenvalues. Beyond the supersymmetric spectral radius, $\gamma^2 > \text{max} (\text{sp}(H_{\text{SUSY}}))$, 
where $\text{sp}(A)$ denotes the spectrum of the operator $A$, the spectrum of $H_{\text{NH}}$ becomes purely imaginary.

We can compactly summarize the above discussion by writing the following expression for the spectrum of $H_{\text{NH}}$,
\begin{align}
    \text{sp}(H_{\text{NH}}) = &\left\{\pm \sqrt{\lambda^2 - \gamma^2}\,|\, \lambda \in \text{sp}(T)\setminus \{0\} \right\} \cup \nonumber\\
& \left\{
\begin{array}{c|l}
s \gamma & s \in \{-, +\} \text{ and } \\
& \ker(T) \cap \mathcal{H}_s \neq \{0\}
\end{array}
\right\},
\end{align}
a structure first identified in \cite{Sutherland1986}. \\

%\noindent
%\textbf{Symmetries:---} 
\subsection{Symmetries}

In this section, we elaborate upon the symmetries present in the non-Hermitian Hamiltonian $H_{\text{NH}}$. We first discuss generalized time-reversal symmetry, then pseudo-Hermiticity and quasi-Hermiticity, and finally particle-hole symmetry. 

%We emphasize that spectral transitions in $H_{NH}$ occur not due to explicit $\mathcal{PT}$-symmetry, but rather due the algebraic relationship between its Hermitian and anti-Hermitian parts. 

In some cases, we can identify a generalized time-reversal symmetry in $H_{\text{NH}}$. Suppose there exists a relabeling of the vertices that swaps their colors while preserving the graph's edge structure; explicitly, such a relabeling is a color-reversing graph automorphism, i.e. a bijection, $\phi:V \to V$, such that $(\phi(v), \phi(w)) \in E$ for all $(v,w) \in E$ and $c(\phi(v)) = -c(v)$ for all $v \in V$. If such a relabeling exists and the corresponding automorphism maps couplings into their complex-conjugates, i.e. $T_{ij} = T_{\phi(i) \phi(j)}^*$, the Hamiltonian $H_{\text{NH}}$ admits a generalized time‑reversal symmetry, $\Theta_\phi$, defined by $(\Theta_\phi \psi)_v := \psi^*_{\phi(v)}$ for all $\psi \in \mathbb{C}^V$. Notably, there exist bipartite graphs that do not have any color‑reversing automorphism. Despite this apparent lack of symmetry, the corresponding non-Hermitian Hamiltonians, $H_{\text{NH}}$, can still exhibit real-to-complex spectral transitions. Thus, the observed real‑to‑complex spectral transition cannot be attributed to traditional $\mathcal{PT}$ or anti‑$\mathcal{PT}$ symmetry.\\ 

For the remainder of this section, we consider an abstracted setting that contains the previously considered Hamiltonians; let $x$ and $y$ denote any pair of anti-commuting operators and let $H = x + i \gamma y$, where $\gamma \in \mathbb{R}$. One choice for $x$ and $y$ could be the matrices $T$ and $\chi$ presented earlier, respectively, in which case $H = H_{\text{NH}}$. Additionally, we assume that $x$ and $y$ are pseudo-Hermitian with the metric operator $\eta$.

Next, we turn our attention to the pseudo-Hermiticity and quasi-Hermiticity of $H$. Whenever $x$ is invertible, $H$ is pseudo-Hermitian and the set $\{\eta x^{m} \, | \, m \in 2 \mathbb{Z}+1 \}$ constitutes a family of indefinite metric operators. If $\eta$ is positive-definite, the choice of indefinite metric $\eta x$ leads to a nice geometric characterization of the eigenstates of $H$: if $\psi$ is a nonzero eigenstate of $H$ with the eigenvalue $\lambda$ such that $\braket{\psi|\eta \psi} = 1$, we deduce the relations 
\begin{align}
    \text{Re}(\lambda) = \braket{\psi|\eta x \psi} &\quad& \text{Im}(\lambda) = \gamma \braket{\psi|\eta y \psi} \label{eqn:indefinite-ev}
\end{align}
by considering the expectation value of $\eta H$ in the state $\psi$. %Thus, $\lambda$ can , the positive (norm) real eigenvalues have positive (negative) indefinite norm with respect to the metric $\eta x$ and the imaginary 

In the following theorem, which was initially proven in \cite[\S 4.1]{Barnett2023PhD}, we deduce criteria that guarantee the quasi-Hermiticity of $H_{\text{NH}}$. When these criteria hold, we find a positive-definite metric operator associated with $H_{\text{NH}}$.  

%Let $\mathfrak{A}$ be a unital $C^*$-algebra, $x$ be a self-adjoint and invertible element of $\mathfrak{A}$, $y = y^* \in \mathfrak{A}$, $y^2 = 1$, $xy + yx = 0$, and $i \gamma \in \mathbb{R}$. Then, the following are equivalent:
\begin{theorem} \label{thm:Quasi-Hermitian}
Let $x, y$ be anti-commuting quasi-Hermitian matrices with the positive-definite metric $\eta$ such that $x$ is invertible and $y^2 = \mathbb{1}$. Additionally, let $\gamma \in \mathbb{R}$ be a real parameter. Then, the following are equivalent:
\begin{enumerate}
    \item $H := x + i \gamma y$ has a real spectrum.
    \item $|\gamma| \leq \min |{\normalfont \text{sp}}(x)|$ %||x^{-1}||^{-1}.
    \item $\eta_{QH} := \mathbb{1} + i \gamma x^{-1} y$ is positive. %More generally, is it true that the index of \eta is # of real eigenvalues of H - # of complex eigenvalues?
\end{enumerate}
For any $\gamma \in \mathbb{R}$, $H$ is pseudo-Hermitian, since $H = x^{-1} H^* x$. Furthermore, $\eta_{QH} H = H^\dag \eta_{QH}$.
\end{theorem}
We present a proof of this theorem in \cref{app:C*-thm-proofs}.

Theorem~\ref{thm:multiplicity} will show that the Hamiltonian, $H$, of \cref{thm:Quasi-Hermitian} has a second-order exceptional point at $|\gamma| = \min|\text{sp}(x)|$. Thus, we conclude that $H$ is diagonalizable with a real spectrum if and only if $H$ is quasi-Hermitian with $\eta_{\text{QH}}$ as a positive-definite metric.

% Lastly, as an application of pseudo-Hermiticity, we establish criteria that confirm whether an eigenvector of the Hamiltonian corresponds to a positive, negative, or imaginary eigenvalue. With an eye towards generalizations and notational simplicity, in the following theorem statement, we define a map that produces expectation values with respect to the state $\psi \in \mathbb{C}^N$ for $N \in \mathbb{N}$. More explicitly, let $\rho_\psi:\mathbb{C}^{N \times N} \to \mathbb{C}$ be such that $\rho_\psi(A) = \braket{\psi|A \psi}$.

% \begin{theorem} \label{thm:Evalue-Separation}
%     Let $x,y$ be anti-commuting quasi-Hermitian matrices with the metric operator $\eta$ and assume $x^2 = \mathbb{1}$. As before, define $H := x + i \gamma y$ where $\gamma \in \mathbb{R}$. Let $\psi$ be an eigenstate of $H$ with the eigenvalue $\lambda$. Then, $\lambda$  
%     \begin{enumerate}
%         \item is real and positive $\, \Leftrightarrow \,$ $\rho_\psi(\eta x) > 0$.
%         \item is real and negative $\, \Leftrightarrow \,$ $\rho_\psi(\eta x) < 0$.
%         \item has a real part equal to zero $\, \Leftrightarrow \,$ $\rho_\psi(\eta x) = 0$.
%     \end{enumerate}
% \end{theorem}

%\noindent
%\textbf{Illustrations:---}
Particle-Hole symmetry \cite{Schomerus2013,Ge2017} is a constraint that forces the spectrum of a Hamiltonian to appear in pairs: whenever $\lambda$ is an eigenvalue, so is $-\lambda^*$. In other words, the complex energy levels are symmetric with respect to a reflection about the imaginary axis. Particle-Hole symmetry is algebraically imposed on the Hamiltonian $H$ via the condition
\begin{align}
    H = -(y \eta)^{-1} H^\dag (y \eta).
\end{align}
When $H^2$ has a real spectrum, we emphasize that $H$ has an additional spectral constraint not present in generic  systems with particle-hole symmetry; the eigenvalues of $H$ are either real or purely imaginary. In contrast, an arbitrary operator with particle-hole symmetry may have eigenvalues whose real and imaginary parts are both nonzero. For example, this can happen when $\eta$ is positive, owing to the quasi-Hermiticity of $H^2$ with respect to the metric $\eta$ \cite{Hill1969}.
\subsection{Illustrations}

In this section, we discuss the properties of four examples of Hamiltonians in our framework, defined in \cref{eqn:HNH-def}. We show their graph structure in \cref{fig:bipartite-graphs-JB} and plot their spectrum in \cref{fig:eval-plots}. These examples showcase how our  framework simplifies existing analysis of some $\mathcal{PT}$-symmetric Hamiltonians while also including Hamiltonians that are not $\mathcal{PT}$-symmetric in a geometric sense.  
% We confirm our results numerically by plotting the eigenvalue trajectories for the array shown in \cref{fig:Bipartite} (c) as a function of the non-Hermitian parameter, $\gamma$, showing the bifurcation across the EP.\\

\begin{figure}[!ht]
    \centering
    \includegraphics[width=0.9\linewidth]{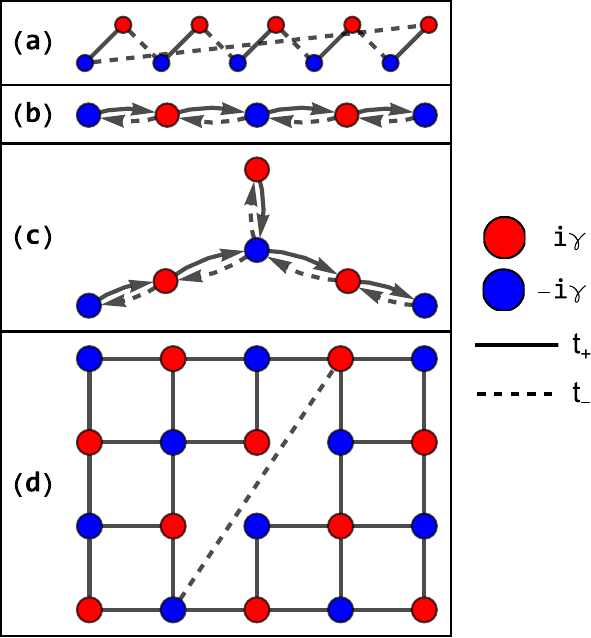}
    \caption{Depiction of four bipartite graphs endowed with non-Hermitian Hamiltonians. Balanced gain, $i \gamma$, and loss, $-i \gamma$, potentials are applied to vertices colored red and blue, respectively. These Hamiltonians admit two distinct nonzero coupling parameters, $t_+$ (solid) and $t_-$ (dashed). 
    These examples showcase how the Hamiltonians in our class may or may not exhibit: non-uniform couplings, as in panels $(a)$ and $(d)$, non-reciprocal couplings, as in panels $(b$–$c)$, or long-range couplings, as in panel $(d)$. The vertices may be elements of spaces with any number of dimensions, highlighted by the graphs in $(c$–$d)$ whose embeddings require more than one dimension. Although panels $(a)$ and $(d)$ depict $\mathcal{PT}$-symmetric dynamics, the Hamiltonians of panels $(b$–$c)$ do not possess a geometric antiunitary symmetry, since no color-reversing weighted graph automorphism exists.}
    \label{fig:bipartite-graphs-JB}
\end{figure}

Panel $(a)$ depicts a Hamiltonian referred to as the complex Su-Schrieffer-Heeger (cSSH) model \cite{Schomerus2013,Lieu2018}. The passive limit, $\gamma \to 0$, was introduced by Su, Schrieffer, and Heeger in \cite{Su1980}. 

When $t_\pm = e^{\pm g}$ for some $g\in \mathbb{R}$ and $\gamma = 0$, panel $(b)$ depicts the celebrated Hatano-Nelson Hamiltonian \cite{HatanoNelson1996Localization} with open boundary conditions. 
The Hatano-Nelson Hamiltonian with open boundary conditions is quasi-Hermitian, as one can verify by constructing a similar Hermitian Hamiltonian \cite{Rozsa1969PeriodicContinuants,williams1969operators}. 
It describes dynamics on a one-dimensional path graph—an example of a tree, i.e. a graph without cycles. In \cref{thm:tree}, we generalize our understanding of the open-boundary Hatano-Nelson Hamiltonian by proving that Hamiltonians with non-reciprocal couplings on arbitrary trees, such as the one in panel $(c)$, are also quasi-Hermitian in the absence of on-site potentials. This quasi-Hermiticity ensures the spectrum of the corresponding supersymmetric Hamiltonian, $H_{\text{SUSY}}$, is real and, consequently, that the non-Hermitian Hamiltonian, $H_{\text{NH}}$, is pseudo-Hermitian. 
\begin{theorem} \label{thm:tree}
    Let $G = (V, E)$ be a connected and directed tree and let $T \in \mathbb{C}^{V \times V}$ be a weighted adjacency matrix of $G$ satisfying $T_{uv} T_{vu} > 0$ for all $(u,v) \in E$. Then, $T$ is quasi-Hermitian. For every vertex, $r \in V$, there exists a unique positive-definite diagonal matrix, $\eta$, that is a metric for $T$ such that $\eta_{rr} = 1$. The matrix elements of this metric are given in \cref{eqn:dipath-metric}.
\end{theorem}
This theorem is proven in \cref{app:tree-proof}. Since the diagonal metric given in \cref{eqn:dipath-metric} commutes with $\chi$, if $T$ is defined as in \cref{thm:tree}, the resulting non-Hermitian Hamiltonian, $H_{\text{NH}}$, respects particle-hole symmetry.

\begin{figure}[!ht]
    \centering
    \includegraphics[width=0.9\linewidth]{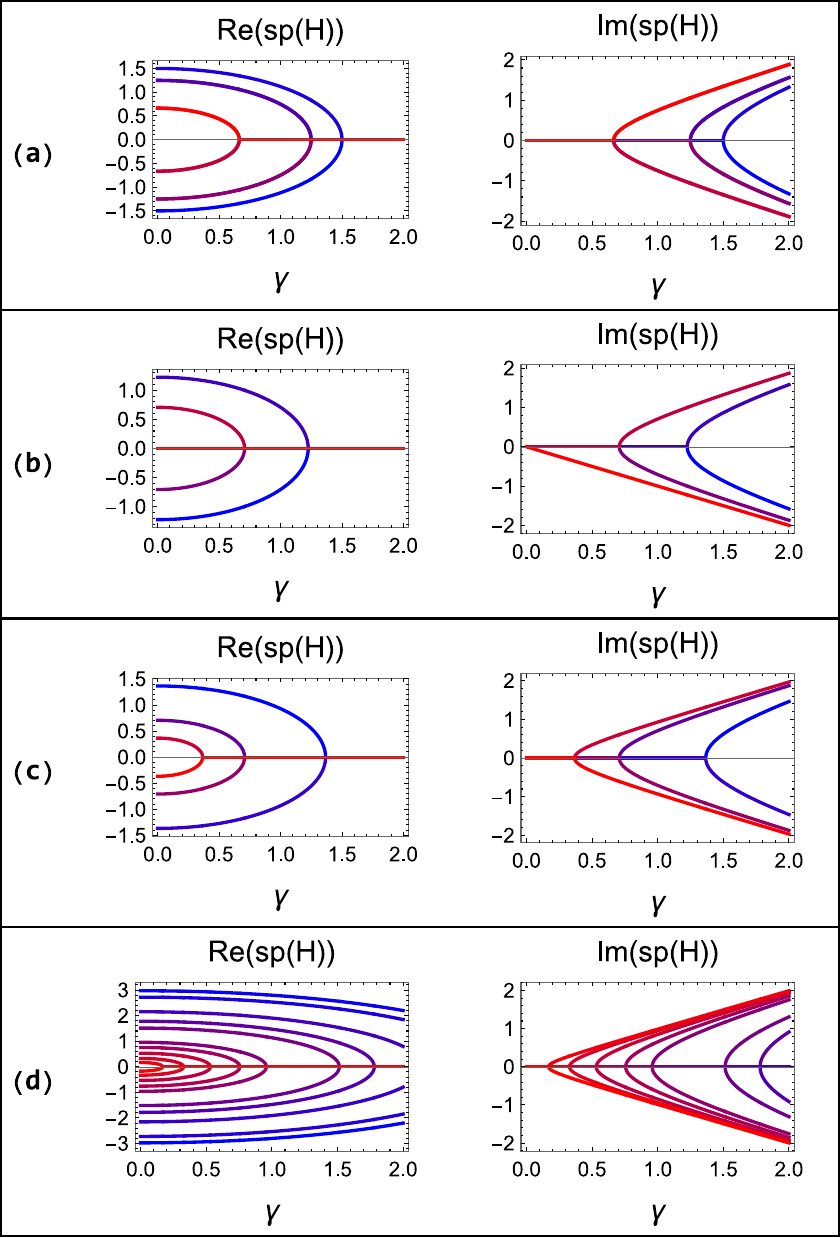}
    \caption{Spectrum of the Hamiltonians with graph structure depicted in \cref{fig:bipartite-graphs-JB}, characterized by the real parameter $\gamma$. The couplings are taken to be $(t_+, t_-) = (1,1/2)$. In all cases, the spectrum exhibits $\mathcal{PT}$-like phase transitions where pairs of real eigenvalues bifurcate into the complex plane in complex-conjugate pairs. As discussed in the main text, exceptional points exist at locations determined by the spectrum of the corresponding passive system, $\gamma = 0$. The threshold of this transition is nonzero if and only if the passive system has no zero modes. One zero mode exists in the passive limit of panel $(b)$; this is because the difference in its number of gain (red) and loss (blue) sites is one, resulting in imbalanced gain and loss and a single purely imaginary eigenvalue without a complex-conjugate counterpart.
    }
    \label{fig:eval-plots}
\end{figure}

\subsection{
\texorpdfstring{$q$-commutation}{q-commutation}
} \label{sec:q-deform}

A Hamiltonian, $H_{\text{NH}}$, of the form defined in \cref{eqn:HNH-def} exhibits second-order exceptional points. In this section, we generalize this construction to include Hamiltonians with $n$-th order exceptional points for any natural number $n$. We start by selecting a primitive $n$-th root of unity, $q \in \mathbb{C}$. Then, the matrix $y$ is said to $q$-\textit{commute} with $x$ if
\begin{align}
    x y = q y x. \label{eqn:q}
\end{align} 
This $q$-deformed commutation relation continuously interpolates between commutation and anti-commutation, which correspond to $q = 1$ and $q = -1$, respectively. We note this kind of deformed commutation relation is familiar to communities studying quantum groups or noncommutative spacetimes, where the quantum plane is an abstract algebra generated by a pair of $q$-commuting elements \cite{Manin1987,Manin2018}.

After fixing two $q$-commuting matrices, $x$ and $y$, we define a Hamiltonian, $H$, that depends on a complex parameter, $\gamma \in \mathbb{C}$, by 
\begin{align}
    H = x + \gamma y. \label{eqn:Ham}
\end{align}
The Hamiltonians adapted to bipartite graphs that were examined earlier correspond to the choice of $x = T$ and $y = i \chi$, where $n = 2$. An example of a Hamiltonian of this form with $n > 2$ is the generalized Hatano-Nelson model of \cite{gohsrich2024exceptional}.

Before addressing the spectrum of $H$ in the general case, we start by displaying some examples of matrices satisfying \cref{eqn:q}. Simple examples of $q$-commuting matrices are Sylvester's \textit{clock} and \textit{shift} matrices \cite{Sylvester1882Nonions},
\begin{align}
    x = {\scriptstyle \begin{bmatrix}
        0 & 1 & 0 & \dots & 0 \\
        0 & 0 & 1 & \dots & 0 \\
        0 & 0 & 0 & \ddots & 0 \\
        \vdots & \ddots & \ddots & \ddots & 1  \\
        1 & \dots & 0 & \ddots & 0
        \end{bmatrix}} &\quad&
    y = {\scriptstyle \begin{bmatrix}
        1 & 0 & 0 & \dots & 0 \\
        0 & q & 0 & \dots & 0 \\
        0 & 0 & q^2 & \ddots & 0 \\
        \vdots & \vdots & \ddots & \ddots & \vdots  \\
        0 & 0 & 0 & \dots & q^{n-1}
        \end{bmatrix}}.
\end{align}
These are a special case of the canonical form for diagonalizable $q$-commuting matrices that was known to \cite{MCINTOSH1962169}. We present this form to the reader using the language of oriented graphs \cite{Sopena2001}.
An \textit{oriented graph} is a directed graph such that if $(u,v)$ is one of its edges, then $(v,u)$ is not an edge; in other words, all edges in an oriented graph are unidirectional. Suppose $G$ is an oriented graph with a finite vertex set that admits an \textit{oriented} $n$\textit{-coloring}, that is defined as a map, $c:V \to \mathbb{C}$, satisfying:
\begin{itemize}
    \item if $(u,v)$ is an edge, then $c(v) = q c(u)$
    \item if $c(u) = c(u')$ and $c(v) = c(v')$, then $(u, v')$ and $(v, u')$ cannot both be edges.
    \item the image of $c$ has $n \in \mathbb{N}$ elements.
\end{itemize}
Let $x$ be a weighted adjacency matrix of this graph and $y$ be the diagonal matrix with the elements $y_{vv} = c(v)$. Then, $x$ and $y$ satisfy \cref{eqn:q}. Sylvester's clock and shift matrices correspond to the case where $G$ is an oriented cycle. An example of an oriented graph is depicted in \cref{fig:oriented}.
\begin{figure}
    \centering
    \includegraphics[width=.9\columnwidth]{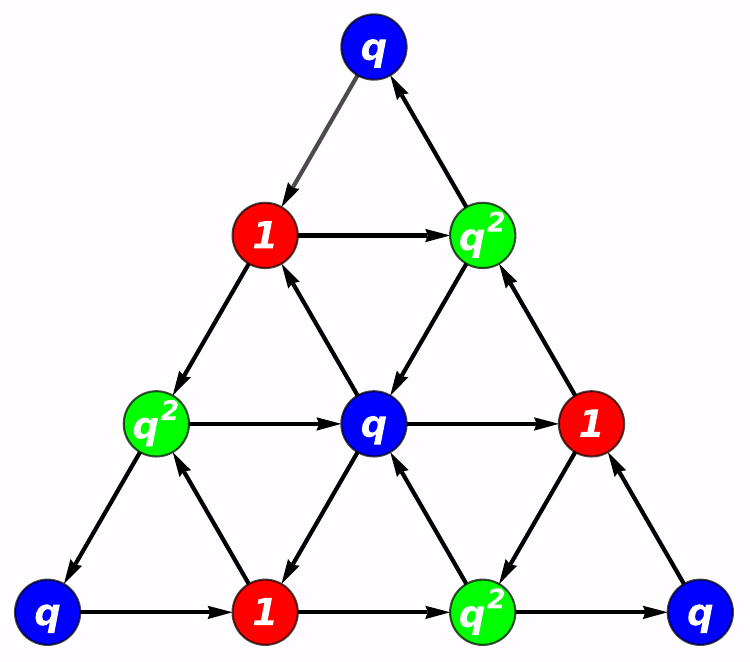}
    \caption{This directed variant of a triangular grid graph is an example of an oriented graph with chromatic number 3. Adapted to this graph is a pair of $q$-commuting matrices: one matrix, $x$, implements nonreciprocal couplings on the plaquettes; the other, $y$ implements an on-site complex potential taking values from the cube roots of unity.}
    \label{fig:oriented}
\end{figure}

In analogy with the preceding analysis for the $n = 2$ case, the task of solving the eigenvalue problem of $H$ is greatly simplified by observing the following instance of the Freshman's dream,
\begin{align}
    H^n &= (x + \gamma y)^n\nonumber \\
    &=x^n + \gamma^n y^n. \label{eqn:Freshman-Dream}
\end{align}
In fact, this instance of the Freshman's dream is a special case of the non-commutative binomial theorem \cite[Thm 5.1]{kac2002quantum} that was proven in \cite{schutzenberger1953interpretation,manin1991notes}, 
\begin{align}
    (x+\gamma y)^m  = \sum_{k = 0}^m \begin{bmatrix}
        m \\ k
    \end{bmatrix}_q (\gamma y)^{m-k} x^k , \label{thm:noncommuting-binomial}
\end{align}
where the \textit{Gaussian binomial coefficients} %for $n \neq 1$?? 
can be defined by a $q$-analog of Pascal's rule \cite{Gauß1808}
\begin{align}
    \begin{bmatrix}
        m \\
        0
    \end{bmatrix}_q &=
    \begin{bmatrix}
        m \\
        m
    \end{bmatrix}_q = 1  \\
    \begin{bmatrix}
        m + 1 \\
        k + 1 
    \end{bmatrix}_q &= \begin{bmatrix}
        m \\
        k+1
    \end{bmatrix}_q + q^{m-k}\begin{bmatrix}
        m \\
        k
    \end{bmatrix}_q, 
\end{align}
where $m, k \in \mathbb{N}$ and $k < m$.
%     \begin{align}
%     \begin{bmatrix}
%         m \\
%         0
%     \end{bmatrix}_q &=
%     \begin{bmatrix}
%         m \\
%         m
%     \end{bmatrix}_q = 1 &\quad& \forall m \in n+1 \\
%     \begin{bmatrix}
%         m \\
%         k
%     \end{bmatrix}_q &= \prod_{j \in k} \frac{q^{m-j} - 1}{q^j - 1} &\quad& \forall k \in m \setminus 1.
% \end{align}
Since $H^n$ is in the span of two commuting objects, $x^n$ and $y^n$, %and the spectral mapping theorem implies $\sigma(H(\gamma)^n) = \sigma(H(\gamma))^n$, 
the spectrum of $H$ satisfies
\begin{align}
    \text{sp}(H)^n \subseteq \{\lambda^n + \gamma^n \mu^n \,|\, (\lambda, \mu) \in \text{sp}(x) \times \text{sp}(y)\}. \label{eqn:Superset-Of-Hamiltonian-Eigenvalues}
\end{align}
If $0 \notin \text{sp}(x) \cup \text{sp}(y)$ and there exists a linear combination of $x^n$ and $y^n$ that yields $\mathbb{1}$, i.e. $\mathbb{1} \in \text{span}\{x^n, y^n\}$, the set inequality of \cref{eqn:Superset-Of-Hamiltonian-Eigenvalues} is an equality.

We conclude this subsection by identifying symmetries that can be present in the spectrum of $H$. If $x$ and $y$ are invertible, a similarity transformation maps $H \to q H$,
\begin{align}
    x y^{-1} H y x^{-1} = q H, \label{eqn:q-commute-H}
\end{align}
in which case the spectrum of $H$ obeys the symmetry $\text{sp}(H) = q \text{sp}(H)$. If, additionally, $n > 2$, and $H$ has a nonzero eigenvalue, then $H$ must have a complex eigenvalue with a nonzero imaginary part.

Suppose $x$ is invertible and pseudo-Hermitian with the metric operator $\eta$. If $(\gamma y) = q^m \eta^{-1} (\gamma y)^\dag \eta$, where $m \in \mathbb{N}$, then $H$ is pseudo-Hermitian with $x^{m} \eta$ as a metric operator. 

\subsection{Eigenspaces} \label{sec:Eigenspaces}

Eigenvalues of matrices are characterized using their multiplicities. Computing multiplicities allows us to assign a precise mathematical meaning to statements regarding whether a Hamiltonian exhibits an exceptional point and, if so, what its order is. We consider two kinds of multiplicities.  If $A$ is an operator with the eigenvalue $\lambda$, the geometric multiplicity of $\lambda$ is 
the maximal number of linearly independent eigenvectors of $A$ with eigenvalue $\lambda$ and is denoted by $\mu_g(\lambda, A)$. The algebraic multiplicity of $\lambda$, which is denoted by $\mu_a(\lambda, A)$, can be defined either as this eigenvalue's multiplicity when considered as a root of the characteristic polynomial of $A$ or as the dimension of the generalized eigenspace corresponding to $\lambda$, as in \cref{eqn:alg-mult-def}. 

In the following theorem, proven in \cref{app:multiplicity-proof}, we characterize the multiplicities of the eigenvalues of the Hamiltonians, $H$, formed as a sum of $q$-commuting matrices, as in \cref{eqn:Ham}. We find that sets of $n$ eigenvectors coalesce into one at certain points in parameter space that we regard as exceptional points with order $n$.
%An eigenvalue is semisimple if $\mu_g(\lambda, A) = \dim(\ker(\lambda - A)^m)$ for all $m \in \mathbb{N}$. 

%Let $y$ be injective, let $\lambda \in \sigma_p(x)$, and let $u \in \ker(\lambda - x)^m \setminus \ker(\lambda - x)^{m-1}$ be a generalized eigenvector of $x$ with rank $m \in \mathbb{N} \setminus \{0\}$. Then, it follows that $y u \in \ker(q \lambda - x)^m \setminus \ker(q \lambda - x)^{m-1}$ is a generalized eigenvector of $x$ associated with the eigenvalue $q \lambda$ with rank $m$.

\begin{theorem} \label{thm:multiplicity}
    Let $\lambda \in \mathbb{C} \setminus \{0\}$ be a nonzero eigenvalue of $x$, suppose $y^n = \mathbb{1}$, and enumerate the $n$-th roots of unity by $\omega_{k} := e^{2 \pi i k/n}$, where $k \in \mathbb{N}$. Then, 
    \begin{enumerate}
    \item 
        $\mu_g(\omega_k \sqrt[n]{\lambda^n + \gamma^n}, H) \geq \mu_g(\lambda, x).$\label{eqn:multiplicity}
    \item If $\gamma^n + \lambda^n = 0$, then
    $\mu_a(0,H) \geq n \mu_g(\lambda, x).$\label{eqn:alg-multiplicity}
    \item The preceding inequalities can be upgraded to equalities under the following circumstances: If either $\gamma = 0$ or $x$ is diagonalizable and invertible, then $\mu_g(\omega_k \sqrt[n]{\lambda^n + \gamma^n}, H) = \mu_g(\lambda, x)$. If $x$ is diagonalizable and invertible, then $\mu_a(0,H) = n \mu_g(\lambda, x) = n \mu_g(0, H)$ when $\gamma^n + \lambda^n = 0$, in which case we call $\gamma$ an $n$-th order exceptional point. \label{item:3}
    \end{enumerate}
    %If $\lambda$ is semisimple and the eigenspace ${\color{red} \ker(\lambda - x)}$ is complemented in ${\normalfont{\text{Dom}}}(x)$, then the inequality in \cref{eqn:multiplicity} can be upgraded to equality.  
\end{theorem}

The condition \(y^{n}=\mathbb{1}\) may be slightly relaxed. If \(y^{n}=\alpha \mathbb{1} \) with \( \alpha \in \mathbb{C} \setminus \{0\}\), then setting 
\(y'=\alpha^{-1/n} y\) defines a matrix that $q$-commutes with $x$ and satisfies 
\((y')^{n}=\mathbb{1}\). 
Applying Theorem~\ref{thm:multiplicity} to \(x\), \(y'\), and a reparametrized Hamiltonian, $H = x + \gamma' y'$ with \(\gamma'=\alpha^{1/n}\gamma\), therefore recovers the essential
conclusions for the original pair \((x,y)\).

\section{Conclusion}
In summary, we have developed a general framework for realizing $\mathcal{PT}$-like phase transitions in non-Hermitian systems without requiring explicit $\mathcal{PT}$-symmetry. By constructing Hamiltonians as nontrivial square roots of energy-shifted supersymmetric partners, we established a systematic route to generate bipartite dynamics with balanced gain and loss, as well as non-reciprocal couplings. The resulting models exhibit second-order exceptional points across which the spectrum undergoes a transition from real to complex eigenvalues. Importantly, similar to standard $\mathcal{PT}$-symmetric systems, the transition between these two spectral regimes is controlled by a single gain/loss parameter. Our approach unifies and generalizes several well-known non-Hermitian models, including the HN and cSSH lattices. Furthermore, the extension to $q$-commuting matrices demonstrates how higher-order exceptional points can be engineered. These results reveal deep connections between non-Hermitian physics and supersymmetry and provide a versatile platform for designing photonic structures with tunable and controllable spectral behavior. Importantly, this approach can be used to engineer non-$\mathcal{PT}$-symmetric Hamiltonians that exhibit real-to-complex eigenvalue transitions and can be implemented with existing photonic and electronic technologies. This capability may open new avenues for designing devices that go beyond what is achievable with conventional $\mathcal{PT}$-symmetric systems.

%In summary, by leveraging the interplay between $q-$analogs, supersymmetry and sublattice symmetry, our approach ensures the controlled realization of $\mathcal{PT}$-like phase transitions. This framework not only establishes conditions for maintaining real eigenvalues within a specific parameter range but also provides a systematic way to predict and manipulate the onset of spectral bifurcations at exceptional points.

%{\color{red} J.B. Do we need to elaborate on this? Should this be moved somewhere else?}
%We note that while the general structure of the Hamiltonian $H_{\text{NH}}$ is closely related to $\mathcal{PT}$-symmetric systems, one easily relate that construction to anti-$\mathcal{PT}$ symmetry by considering the Wick-rotated Hamiltonian $\bar{H}_{\text{NH}}=iH_{\text{NH}}=iT +\gamma \chi$.  

%Olga Taussky also remarked that whenever we have a pair of "interesting" matrices, we should study their algebra \cite{Taussky1988}. The algebra generated by a pair of $q$-commuting matrices has been called the \textit{quantum plane} \cite{Manin2018} and been subjected to detailed investigation, it would be interesting to know whether other curves in the quantum plane exhibit exceptional points.

\section{Acknowledgments}

J.B. is grateful to Aron C. Wall, who showed him an example of a Hamiltonian in the class defined by \cref{eqn:HNH-def} and a corollary of \cref{eqn:indefinite-ev} in the context of his work on holography \cite{AraujoRegado2023}.\\
J.B. is supported by the Basque government through the BERC 2022-2025 program. J. B. is also supported by the Grant PID2024-156184NB-I00 funded by MICIU/AEI/10.13039/501100011033 and cofunded by the European Union.

R.E. acknowledges support from the AFOSR Multidisciplinary University Research Initiative Award on Programmable Systems with Non-Hermitian Quantum Dynamics (Grant No.FA9550-21-1-0202).

\appendix
\section{\texorpdfstring{$2 \times 2$ Matrices}{2x2 Matrices}} \label{app:2x2}
This appendix presents a proof of \cref{thm:2x2}.
\begin{proof}
We start by introducing some notation.
\begin{itemize}
\item Given a complex matrix, $H$, we denote its Hermitian and anti-Hermitian parts by $H_+$ and $H_-$, respectively. Explicitly, 
\begin{align}
    H_+ := \frac{H + H^\dag}{2}, &\quad& 
    H_- := \frac{H - H^\dag}{2 i} .
\end{align}
\item We let $\mathfrak{sl}_2(\mathbb{C})$ denote the complex vector space whose elements are traceless $2 \times 2$ matrices. 
The \textit{Pauli basis} of $\mathfrak{sl}_2(\mathbb{C})$ is a basis comprised of three elements \cite{Pauli1927},
\begin{align}
    \vec{\sigma}_1 &:= \begin{pmatrix}
        0 & 1 \\
        1 & 0 
    \end{pmatrix}, \,\, 
    \vec{\sigma}_2 := \begin{pmatrix}
        0 & -i \\
        i & 0 
    \end{pmatrix} \,\, 
    \vec{\sigma}_3 := \begin{pmatrix}
        1 & 0 \\
        0 & -1
    \end{pmatrix}.
\end{align}
If $(\vec{\alpha}_1, \vec{\alpha}_2, \vec{\alpha}_3)$ is a 3-tuple of real numbers, its corresponding \textit{Pauli vector} is the matrix $\vec{\alpha} \cdot \vec{\sigma} = \sum_{i = 1}^3 \vec{\alpha}_i \vec{\sigma}_i$. 
\item For every $H \in \mathfrak{sl}_2(\mathbb{C})$, we define a pair of real 3-tuples, $\vec{\alpha}_{\pm}(H)$, by 
$\vec{\alpha}_{\pm}(H)_i := \frac{1}{2} \text{Tr}(H_{\pm} \vec{\sigma}_i).$ The Hermitian and anti-Hermitian parts of $H$ are $H_{\pm} = \vec{\alpha}_{\pm}(H) \cdot \vec{\sigma}$, respectively.
\end{itemize}

Using a product identity for Pauli vectors,
\begin{align}
    (\vec{\alpha}\cdot \vec{\sigma} )(\vec{\beta} \cdot \vec{\sigma}) = \vec{\alpha}\cdot \vec{\beta} + i (\vec{\alpha} \times \vec{\beta}) \cdot \vec{\sigma},
\end{align}
the anti-commutator of $H_+$ and $H_-$ evaluates to
\begin{align}
   H_+ H_- + H_- H_+ = 2 \vec{\alpha}_+(H) \cdot \vec{\alpha}_-(H),
\end{align}
implying that the Hermitian and anti-Hermitian parts of $H$ anti-commute if and only if the vectors $\vec{\alpha}_{\pm}(H)$ are orthogonal. Thus, \cref{thm:2x2} is logically equivalent to the following statement: a $2 \times 2$ complex matrix, $H$, is pseudo-Hermitian if and only if $\vec{\alpha}_{+}(H) \cdot \vec{\alpha}_-(H) = 0$. This equivalent statement is known to be true; we summarize one proof for completeness, following \cite[\S 2.7]{Barnett2023PhD}.

Suppose $H \in \mathfrak{sl}_2(\mathbb{C})$ is pseudo-Hermitian. Then, the characteristic polynomial of $H$ has real coefficients \cite{Bender2010}. Using the well-known formula for the determinant of a sum of $2 \times 2$ matrices, $A, B \in \mathbb{C}^{2 \times 2}$, \cite[p. 56]{PopFurdui_SquareMatrices2}
\begin{align}
    \det(A+B) &= \det(A) + \det(B) \nonumber \\
    &+ \text{Tr}(A) \text{Tr}(B) - \text{Tr}(AB),
\end{align}
we can express the characteristic polynomial of $H$ as
\begin{align}
    \det(\lambda - H) = \lambda^2 &- ||\vec{\alpha}_{+}(H)||^2 + ||\vec{\alpha}_-(H)||^2 \nonumber \\ &- 2i \vec{\alpha}_+(H) \cdot \vec{\alpha}_-(H).
\end{align}
Thus, $H$ is pseudo-Hermitian only if $\vec{\alpha}_+(H) \cdot \vec{\alpha}_-(H) = 0$.

Conversely, assume that $\vec{\alpha}_+(H) \cdot \vec{\alpha}_-(H) = 0$. If $\vec{\alpha}_+(H) = 0$, then $H$ is anti-Hermitian and, thus, pseudo-Hermitian. If $\vec{\alpha}_+(H) \neq 0$, then $\vec{\alpha}_+(H)\cdot \vec{\sigma}$ is a Hermitian and invertible metric operator rendering $H$ pseudo-Hermitian. 
\end{proof}

% \section{Arbitrary Supersymmetric case}

% {\color{red} I had a goal of writing about how every shifted supersymmetric Hamiltonian that's energy shifted admits a square root of the form we talk about. I don't believe it to be necessary.}

%Eigenvectors in this basis:

%$\gamma \psi_{j+1} = ((\lambda^n + \gamma^n)^{1/n}-q^j)\psi_j$
\section{Generalizations} \label{app:generalization}
In the main text, we restricted our discussion to finite-dimensional complex matrices for clarity of exposition. Here, we outline how the principal results extend to more general algebraic frameworks.

\begin{itemize}
\item When the vertex set $V$ has countably-infinite cardinality, the space $\mathbb{C}^V$ is too large to be a Hilbert space. For the most part, the discussion above carries through when the spaces $\mathbb{C}^V$ are replaced by the subset of square summable sequences, 
\begin{align}
    l_2(V) := \{\psi \in \mathbb{C}^V \,|\, \sum_{v \in V} |\psi_v|^2 < +\infty \},
\end{align} which is a Hilbert space. In this case, to avoid potential complications with operator domains, we assume that the maps $T_\pm$ are bounded. This requires all graphs in consideration to have a bounded vertex degree, 
\begin{align}
    \sup_{v \in V}|\{u \in V \,|\, (v, u) \in E \text{ or } (u, v) \in E\}| \in \mathbb{N}.
\end{align}
We note that a finite operating regime for the parameter $\gamma$ where $H_{\text{NH}}$ has a real spectrum only exists if $T$ is bounded below. 

\item Theorem~\ref{thm:Quasi-Hermitian} and \cref{eqn:indefinite-ev} can be abstracted to the setting of $C^*$-algebras, as explained in Appendix~\ref{app:C*-thm-proofs}.
\item The discussion in \cref{sec:q-deform} generalizes straightforwardly to the case where $x$ and $y$ are elements of a Banach algebra.
\item In \cref{thm:multiplicity}, the inequalities of \cref{eqn:multiplicity,eqn:alg-multiplicity} apply to the case where $x$ and $y$ are bounded operators acting on an infinite-dimensional Hilbert space. The proof provided in \cref{app:multiplicity-proof} for \cref{item:3} requires this Hilbert space to be finite-dimensional. 

%In analogy with the case of matrices, if $\mathfrak{A}$ is a $C^*$-algebra, one of its elements $x \in \mathfrak{A}$ is defined to be \textit{pseudo-Hermitian} if and only if there exists an inner automorphism, $\cdot \to \eta^{-1} \cdot \eta$ that maps $x$ into $x^*$. If there exists such an automorphism where $\eta$ is strictly-positive, then $x$ is called \textit{quasi-Hermitian}.

\end{itemize}

\section{\texorpdfstring{Proof of Theorem~\ref{thm:Quasi-Hermitian}}{Proof of Theorem 2}
} \label{app:C*-thm-proofs}

We will prove \cref{thm:Quasi-Hermitian} in an abstracted setting that is more general than the version stated in the main text. Namely, we consider \textit{unital $C^*$-algebras}, which include algebras of finite-dimensional complex matrices as a special case. Notably, there exist $C^*$-algebras that do not admit finite-dimensional matrix representations, such as the algebra of bounded operators acting on an infinite-dimensional Hilbert space. 

To keep the writing here self-contained, we summarize several standard definitions from the theory of $C^*$-algebras. For a more detailed exposition, the reader is invited to review textbooks such as \cite{Doran2018,Bru2023}. 
An \textit{associative algebra}, $\mathfrak{A}$, is a complex vector space equipped with an associative bilinear product, $\cdot:\mathfrak{A} \times \mathfrak{A} \to \mathfrak{A}$, that distributes over addition. Given two elements of an associative algebra, $x, y \in \mathfrak{A}$, we use the abbreviation $x y := \cdot((x, y))$. A \textit{Banach algebra} is an associative algebra, $\mathfrak{A}$, that is also a complete normed space where the \textit{Banach product inequality} holds, namely $||x y|| \leq ||x|| \cdot ||y||$ for all $x, y \in \mathfrak{A}$. A ${}^*$-\textit{algebra} is an associative algebra, $\mathfrak{A}$, equipped with a map, ${}^*:\mathfrak{A} \to \mathfrak{A}$, called the \textit{involution}, which is an antiautomorphism satisfying $(x^*)^* = x$ for all $x \in \mathfrak{A}$. A \textit{unital} algebra, $\mathfrak{A}$, is an associative algebra containing a two-sided identity, which is an element $\mathbb{1} \in \mathfrak{A}$ such that $\mathbb{1} x = x  \mathbb{1} = x$ for all $x \in \mathfrak{A}$. A $C^*$-\textit{algebra} is an associative Banach ${}^*$-algebra where the $C^*$-\textit{identity} holds, namely 
$||x^* x|| = ||x||^2$ for all $x \in \mathfrak{A}.$

An element of a ${}^*$-algebra, $x \in \mathfrak{A}$, is called \textit{self-adjoint} if $x = x^*$. An element of a unital algebra, $a \in \mathfrak{A}$, is called \textit{invertible} if there exists an element, $a^{-1}$, such that $a  a^{-1} = a^{-1} a = \mathbb{1}$. The \textit{spectrum} of an element of a unital associative algebra, $a \in \mathfrak{A}$, is the set \begin{align}
    \text{sp}(a) := \{\lambda \in \mathbb{C} \,|\, \lambda \mathbb{1} - a \text{ is not invertible} \}.
\end{align}
A self-adjoint $C^*$-algebra element is called \textit{positive} if its spectrum does not contain a negative real number.
The quintessential example of a finite-dimensional $C^*$-algebra is the set of $N \times N$ complex matrices, with $N \in \mathbb{N}$. In this example, the norm is the operator norm, the involution is complex-conjugate transposition in the canonical basis, and the spectrum of a matrix is the set of its eigenvalues. 

Let $x, y$ be elements of a unital $C^*$-algebra, $\mathfrak{A}$, such that $x$ is invertible and $y^2 = \mathbb{1}$. Additionally, assume the existence of an invertible and positive element, $\eta \in \mathfrak{A}$, such that $z = \eta^{-1} z^* \eta$ for all $z \in \{x, y\}$. Then, we prove that the assertions of \cref{thm:Quasi-Hermitian} hold for these abstract algebra elements.

\begin{proof}[Proof of \Cref{thm:Quasi-Hermitian}] 
First, we argue that we can take $\eta = \mathbb{1}$ without loss of generality. To understand why, we may regard $x$ and $y$ as self-adjoint elements of a different $C^*$ algebra, $\mathfrak{A}_\eta$, whose elements and associative algebra operations are taken from $\mathfrak{A}$ and whose the involution and norm, ${\cdot}^{*_\eta}$ and $||\cdot||_\eta$, are defined by \cite{Karuvade2022}
\begin{align}
   z^{*_\eta} &:= \eta^{-1} z^* \eta &\quad& \forall z \in \mathfrak{A}_\eta \\
   ||z||_\eta &:= ||\eta^{1/2} z \eta^{-1/2}|| &\quad& \forall z \in \mathfrak{A}_\eta,
\end{align}
where $\eta^{1/2}$ refers to the unique positive square root of $\eta$. Suppose we have shown that the theorem is true for an arbitrary $C^*$-algebra with $\eta = \mathbb{1}$. Then, it must be true for the algebras $\mathfrak{A}_{\eta'}$ for any positive $\eta'$, which implies the theorem statement in $\mathfrak{A}$ if we select $\eta' = \eta$.

We note that identity $\eta_{\text{QH}} H = H^* \eta_{\text{QH}}$ can be deduced from $H = x^{-1} H^* x$ and $\eta_{\text{QH}} = x^{-1} H$. This is a specific instance of the procedure for generating metric operators outlined in \cite{bian2019time}.

The spectrum of $H$ can be characterized by using \cref{eqn:Superset-Of-Hamiltonian-Eigenvalues} and observing its particle-hole  symmetry. Explicitly,
\begin{align}
    \text{sp}(H) = \{ \pm \sqrt{\mu^2 - \gamma^2} \, |\, \mu \in \text{sp}(x)\},
\end{align}
from which the equivalence
$(1) \Leftrightarrow (2)$ becomes readily apparent. 

We now discuss the claim $(2) \Leftrightarrow (3)$. Note that the self-adjoint operator $\eta_{\text{QH}} - \mathbb{1}$ admits a chiral symmetry, since $x (\eta_{\text{QH}} - \mathbb{1}) = (\eta_{\text{QH}}-\mathbb{1}) x$. Thus, $\eta_{\text{QH}}$ is positive if and only if the spectrum of $i \gamma x^{-1} y$ is a subset of the unit disk. We can determine bounds on the spectrum of an arbitrary $C^*$-algebra element, $a$, using their norms \cite[\S 4.2]{Bru2023},
\begin{align}
||a|| = \sup\{|\lambda|\,| \lambda \in \text{sp}(a) \}. \label{eqn:norm-sp-relation}
\end{align} 
In this case, the norm of $i \gamma x^{-1} y$ can be related to the norm of $x^{-1}$ by using the $C^*$-identity,
\begin{align}
    ||x^{-1} y|| &= \sqrt{||(x^{-1} y)^* ( x^{-1} y)||}\\
    &= \sqrt{||x^{-1} y y x^{-1}||} \\
    &= \sqrt{||x^{-2}||} \\
    &\Downarrow \nonumber \\
    ||y x^{-1}|| &= ||x^{-1}||, \label{eqn:Norm-of-x-inverse}
\end{align}
thereby implying that $\eta_{\text{QH}}$ is positive if and only if $|\gamma|  \leq ||x^{-1}||^{-1}$. We can recast this equivalence as $(2) \Leftrightarrow (3)$ by noting
\begin{align}
    ||a^{-1}||^{-1} = \inf(|\text{sp}(a)|),\label{eqn:spectral-inf}
\end{align} 
which follows from the spectral mapping theorem, and $\inf(|\text{sp}(a)|) = \min(|\text{sp}(a)|)$, which follows from the compactness of the spectrum.
\end{proof}

We conclude this appendix by generalizing \cref{eqn:indefinite-ev} to $C^*$-algebraic settings. This requires an abstract definition of an eigenstate; here, we use and summarize the definition of \cite{DeNittis2023}. In a $C^*$-algebra, $\mathfrak{A}$, states are defined as suitable elements of the continuous dual space of $\mathfrak{A}$, which is the set of all bounded linear maps from $\mathfrak{A}$ to $\mathbb{C}$ and is denoted by $\mathcal{B}(\mathfrak{A},\mathbb{C})$. A \textit{state} on $\mathfrak{A}$ is a positive normalized linear functional. Explicitly, the set of states on $\mathfrak{A}$ is 
\begin{align}
    S(\mathfrak{A}) := \bigcap_{x \in \mathfrak{A}} \left\{\varphi \in \mathcal{B}(\mathfrak{A}, \mathbb{C}) \,|\, \varphi(x^* x) \geq 0 \wedge ||\varphi|| = 1 \right\}.
\end{align}
The set of \textit{eigenstates} of a $C^*$-algebra element, $a \in \mathfrak{A}$, is the set of states that satisfy an algebraic version of the eigenstate condition,
\begin{align}
    \text{Eig}_\lambda(a) := \bigcap_{x \in \mathfrak{A}} \left\{ \omega \in S(\mathfrak{A})\,|\, \omega(x a) = \lambda \omega(x) \right\}.
\end{align}
If $\text{Eig}_\lambda(a) \neq \{0\}$, then $\lambda$ is referred to as an \textit{eigenvalue} of $a$. 

Let $\rho \in \text{Eig}_\lambda(x + i \gamma y)$, where $x$ and $y$ are anti-commuting pseudo-Hermitian elements, meaning there exists an invertible $\eta = \eta^* \in \mathfrak{A}$ such that $\eta x = x^* \eta$ and $\eta y = y^* \eta$. Then, 
\begin{align}
    \text{Re}(\lambda) \rho(\eta) = \rho(\eta x) &\quad& \text{Im}(\lambda) \rho(\eta) = \gamma \rho(\eta y)
\end{align}
is a $C^*$-algebraic generalization of \cref{eqn:indefinite-ev}.

\section{\texorpdfstring{Proof of \cref{thm:tree}}{Proof of theorem 3}} \label{app:tree-proof}
In the following two appendices, we define the natural numbers as sets in Zermelo–Fraenkel set theory, where a number, $n \in \mathbb{N}$, is identified with the set of all lesser numbers, $n = \{m \in \mathbb{N} \,|\,m < n\}$.

\begin{proof}
Suppose $\eta \in \mathbb{C}^{V \times V}$ is a diagonal metric operator for $T$ with $\eta_{rr} = 1$. Then, the definition of pseudo-Hermiticity, \cref{eqn:pseudo-defn}, implies
\begin{align}
    \eta_{uu} T_{uv} = T^*_{vu} \eta_{vv} &\quad& \forall u, v \in V.\label{eqn:diagonal-recurrence}
\end{align} 
This family of equations admits a unique solution that can be found by considering them as recurrence relations along paths in the tree.
Since $G$ is a connected tree, given $u \in V$, there exists a unique simple path, $P_{u}$, that walks from $r$ to $u$. Explicitly, $P_{u}$ is defined by the following properties:
    \begin{itemize}
        \item $P_{u}$ is a \textit{directed path}. This means there exists a bijective enumeration of its vertices, $v_{i}$ where $i \in d$ for some $d \in \mathbb{N}$, such that $e$ is an edge of $P_{uv}$ if and only if $e = (v_i, v_{i+1})$ for some $i \in d$.
        \item Both $u$ and $r$ belong to exactly one edge of $P_{u}$. The edge containing the root is of the form $(r,v)$ for some $v \in V$.
    \end{itemize}
     Let $\delta_{uv}$ denote the Kronecker delta. Then, the recurrence relation \cref{eqn:diagonal-recurrence} implies that the matrix elements of $\eta$ are
    \begin{align}
        \eta_{uv} := \delta_{uv} \prod_{(x, y) \in P_u} \frac{T_{yx}^*}{T_{xy}}. \label{eqn:dipath-metric}
    \end{align} 
    Positive-definiteness of $\eta$ follows from the constraint $T_{uv} T_{vu} > 0$ for all $(u, v) \in E$.
\end{proof}

\section{\texorpdfstring{Proof of \cref{thm:multiplicity}}{Proof of theorem 4}} \label{app:multiplicity-proof}

Before proceeding with the proof, we summarize some elementary aspects of linear algebra.
The \textit{eigenspace} of an operator, $A$, corresponding to the eigenvalue $\lambda$ is the set $\ker (\lambda \mathbb{1} - A)$, where $\ker$ denotes the kernel of an operator. The multiplicities of the operator's eigenvalues are explicitly defined by
\begin{align}
    \mu_g(\lambda, A) &:= \dim \ker(\lambda \mathbb{1} - A),\\
    \mu_a(\lambda, A) &:= \sup_{N \in \mathbb{N}} \dim \ker (\lambda \mathbb{1} - A)^N. \label{eqn:alg-mult-def}
\end{align}
A \textit{Jordan chain} with length $d \in \mathbb{N}$ for the eigenvalue $\lambda$ of $A$ is a sequence of nonzero vectors, $w:d \to \text{Dom}(A^d) \setminus \{0\}$, such that 
\begin{align}
    (A - \lambda) w_j &= w_{j+1} &\quad& \forall j \in d-1 \label{eqn:Jordan1} \\
    (A - \lambda) w_{d-1} &= 0. \label{eqn:Jordan2}
\end{align}
If $w(k)$ is an indexed family of Jordan chains with lengths $d_k$ such that their eigenvectors, $w(k)_{d_k-1}$, are linearly independent, then the union of all vectors in these chains is linearly independent and, consequently, 
\begin{align}
    \mu_a(\lambda, A) \geq \sum_k d_k. \label{eqn:chain-sum}
\end{align}

We start by proving the following lemma.
\begin{lemma} \label{lemma:Linear-Independence}
    Let $x$ and $y$ be matrices such that $x y = q y x$ for some $q \in \mathbb{C}$ and let $\lambda \in \mathbb{C} \setminus \{0\}$. Assume there exists a positive integer, $n \in \mathbb{N} \setminus \{0\}$, such that for all $j \in n$, $q^j = 1 \, \Leftrightarrow j = 0$ and $\ker(y^j) \cap \ker(\lambda \mathbb{1} - x) = \{0\}$. Let $m \in \mathbb{N}$ and let $v_{k \in m}$ be a sequence of $m$ linearly-independent eigenvectors of $x$ with eigenvalue $\lambda$. Then, $\{y^j v_k \,|\, k \in m, j \in n\}$ is linearly independent.
\end{lemma}
\begin{proof}
Suppose $c \in \mathbb{C}^{n \times m}$ is such that 
    \begin{align}
        \sum_{j \in n, k \in m} c_{jk} y^j v_k = 0.
    \end{align}
Then, the set $\{y^j v_k \,|\, k \in m, j \in n\}$ is linearly independent if $c = 0$ necessarily holds. To understand why this is the case, let 
\begin{align}
    \tilde{v}_j := y^j \sum_{k \in m} c_{jk} v_k. \label{eqn:vTilde-defn}
\end{align}
Then, since $y$ $q$-commutes with $x$, it follows that $\tilde{v}_j$ is an element of the eigenspace of $x$ with eigenvalue $q^j \lambda$, or more explicitly, 
\begin{align}
x \tilde{v}_j = q^j \lambda \tilde{v}_j.    
\end{align}
Because a set of nonzero eigenvectors corresponding to distinct eigenvalues is linearly independent, the constraint $\sum_{j \in n} \tilde{v}_j = 0$ implies $\tilde{v}_j = 0$ for all $j \in n$. Thus, by \cref{eqn:vTilde-defn}, we find 
\begin{align}
    \sum_{k \in m} c_{jk} v_k &\in \ker(y^j) \cap \ker(\lambda \mathbb{1} - x)\\
    &\Updownarrow \nonumber \\
    \sum_{k \in m} c_{jk} v_k &= 0.
\end{align}
Since $v_k$ was assumed to be linearly independent, it must be the case that $c_{jk} = 0$ for all $j\in n$ and $k\in m$. 
\end{proof}

We proceed by writing a proof of \cref{thm:multiplicity}.

\begin{proof}[Proof of \Cref{thm:multiplicity}]
We start by addressing the case of $\gamma = 0$, where $H = x$. Since $y^{-1} x y = q x$, $v \in \ker(\lambda \mathbb{1} - x)^N$ if and only if $y v \in \ker(q \lambda \mathbb{1} - x)^N$ for all $N \in \mathbb{N}$. Thus, the algebraic and geometric multiplicities of $\lambda$ in $x$ coincide with the algebraic and geometric multiplicities of $q^r \lambda$ in $x$ for all $r \in \mathbb{N}$, respectively. Since there exists a value of $r$ such that $q^r = \omega_k$, the algebraic and geometric multiplicities in $x$ of $\lambda$ and $\omega_k \lambda$ are equal for all $\lambda \in \mathbb{C}$, thereby proving the theorem for $\gamma = 0$. Thus, moving forward, we only need to consider the case $\gamma \neq 0$.

We prove the assertions of the theorem in order.
\begin{enumerate}
    \item We proceed by systematically constructing eigenvectors of $H$ corresponding to eigenvectors of $x$. To be precise, we find a one-to-one linear map, $T_{k, \lambda}$, that maps the eigenspace of $x$ with eigenvalue $\lambda$ to the eigenspace of $H$ with eigenvalue $\omega_k \sqrt[n]{\lambda^n + \gamma^n}$. Since the geometric multiplicity of an eigenvalue is the dimension of its eigenspace and since $\dim V \leq \dim W$ whenever $V$ and $W$ are vector spaces such that there exists a one-to-one linear map from $V$ to $W$, the inequality of \cref{eqn:multiplicity} immediately follows from the existence of $T_{k,\lambda}$.
    
    Let $T_{k, \lambda}: \ker(\lambda \mathbb{1} - x) \to \text{Dom}(x)$, be defined by
    \begin{align}
        T_{k, \lambda} v := \sum_{j \in n} \alpha_{j;k} y^j v,
    \end{align}
    where the sequence $\alpha: n\times n \to \mathbb{C}$ is defined by
    \begin{align}
        \alpha_{0;k} &:= 1\\
        \alpha_{j+1;k} &:= \dfrac{\gamma \alpha_{j;k}}{-\lambda q^{j+1} + \omega_k \sqrt[n]{\lambda^n + \gamma^n}}, \label{defn:alpha}
    \end{align}
    with $k \in n$ and $j \in n-1$. Note the denominator in the construction for $\alpha$ is nonzero because we assumed $\gamma \neq 0$. 
    
    We now prove that $T_{k, \lambda}$ is one-to-one by contradiction. A linear map is one-to-one if and only if its kernel is the singleton set whose only element is the zero vector. Suppose there exists a nonzero vector, $v$, in the kernel of $T_{k, \lambda}$. By definition of $T_{k, \lambda}$, this would imply that the set $\{y^j v\,|\, j \in n\}$ is linearly dependent, which contradicts lemma~\ref{lemma:Linear-Independence}. %see e.g. https://math.stackexchange.com/questions/29371/how-to-prove-that-eigenvectors-from-different-eigenvalues-are-linearly-independe

Next, we show that elements of the image of $T_{k, \lambda}$ are eigenvectors of $H$. Note 
    \begin{align}
        H T_{k, \lambda} v &= \sum_{j \in n} \alpha_{j;k} (x + \gamma y) y^j v\\
        &= \sum_{j \in n} \alpha_{j;k}  (\lambda q^j  y^j + \gamma y^{j+1}) v \\
        &= \sum_{j=1}^{n-1} \left(q^j \lambda \alpha_{j;k} + \gamma \alpha_{j-1;k} \right) y^j v \nonumber \\
        &+ (\lambda + \gamma \alpha_{n-1;k}) v. \label{eqn:HTv}
    \end{align}
We proceed by producing a closed-form expression for $\alpha_{n-1;k}$. First, note 
\begin{align}
    \alpha_{n-1;k}
    &= \gamma^{n-1} \left(-\lambda + \omega_k \sqrt[n]{\lambda^n + \gamma^n}\right) \nonumber \\
    &\times\prod\limits_{j = 0}^{n-1} (-\lambda q^j + \omega_k \sqrt[n]{\lambda^n + \gamma^n})^{-1} \label{eqn:alpha-prod}.
\end{align}
This product can be re-expressed as a sum using Gauss' $q$-\textit{binomial theorem}, which states for every $a, b \in \mathbb{C}$ \cite{Gauß1808,kac2002quantum},
\begin{align}
\prod_{k \in m} (a + q^k b) = \sum_{j = 0}^n q^{j(j-1)/2} \begin{bmatrix}
        m \\
        j
    \end{bmatrix}_q
    a^{n-j} b^j.
\end{align}
In summary,
\begin{align}
    \alpha_{n-1;k} &= \gamma^{-1} \left(-\lambda + \omega_k \sqrt[n]{\lambda^n + \gamma^n}\right).
    \label{eqn:alpha-last}
\end{align}
Equations~\ref{defn:alpha} and \ref{eqn:alpha-last} can be used to rewrite \cref{eqn:HTv} as
\begin{align}
    H\, T_{k, \lambda} v = \omega_k \sqrt[n]{\lambda^n + \gamma^n}\, T_{k, \lambda} v,
\end{align}
which means that $T_{k, \lambda} v$ resides in the eigenspace of $H$ with the eigenvalue $\omega_k \sqrt[n]{\lambda^n + \gamma^n}$.

\item Next, let $\text{sp}_p(x)$ denote the set of eigenvalues of $x$ and consider the following subset
\begin{align}
    E = \{z \in \text{sp}_p(x) \,|\, z^n + \gamma^n = 0\}.
\end{align}
Intuitively, the eigenvalues of $x$ in the set $E$ are those that merge when $\gamma$ is an exceptional point of $H$. 

Suppose $\lambda \in E$ and let $\{v_k \,|\, k \in \mu_g(\lambda, x)\}$ be a basis of the eigenspace of $x$ with eigenvalue $\lambda$. Then, in the remainder of this paragraph, we show that the sequences 
\begin{align}
    w(k)_j := H^j v_k&\quad& (j \in n)
\end{align}
are linearly-independent Jordan chains for the eigenvalue $0$ in $H$ with length $n$. To do this, we need to verify three properties: that $w(k)_j$ is nonzero, satisfies \cref{eqn:Jordan1}, and satisfies \cref{eqn:Jordan2}. Equation~\ref{eqn:Jordan1} holds by definition. The other two properties are equivalent to the assertion that if $m \in \mathbb{N}$, then $H^m v_k = 0 \, \Leftrightarrow \, m \geq n$. This assertion follows from the noncommutative binomial theorem of \cref{thm:noncommuting-binomial} \cite{schutzenberger1953interpretation}, which implies
\begin{align}
    H^m v_k = \sum_{j = 0}^m \gamma^{m-j} \lambda^j \begin{bmatrix}
        m \\
        j
    \end{bmatrix}_q y^{m-j} v_k.
\end{align}
When $m < n$, $H^m v_k \neq 0$ follows from the previously-established linear independence of the set $\{y^j v \,|\, j \in n\}$. The identity $H^n v_k = 0$ follows from the Freshman's dream, \cref{eqn:Freshman-Dream}. Finally, we remark that the linear independence of the eigenvectors $H^{n-1} v_k$ is a corollary of lemma~\ref{lemma:Linear-Independence}.

The inequality of \cref{eqn:alg-multiplicity} follows by applying \cref{eqn:chain-sum} to the set of $\mu_g(\lambda, x)$ linearly-independent Jordan chains with length $n$, $w(k)$, constructed in the preceding paragraph.

\item Below, we assume that $x$ is diagonalizable, invertible, and $\text{Dom}(x)$ is the finite-dimensional space $\mathbb{C}^N$ with $N \in \mathbb{N}$. The diagonalizability of $x$ is equivalent to asserting 
\begin{align}
N = \sum_{\lambda \in \text{sp}(x)} \mu_g(\lambda, x). \label{eqn:contradiction-in-mult}
\end{align}
By the fundamental theorem of algebra, the sum of algebraic multiplicities of $x$ is $N$, implying 
\begin{align}
    N
    &= \mu_a(0,H) + \sum_{\mu \in \text{sp}(H) \setminus \{0\}} \mu_a(\mu, H). \label{eqn:mult-sum}
\end{align}
Let us consider the two terms in this sum separately. We first note that
\begin{align}
    \mu_a(0,H) \geq  \sum_{\lambda\in E} \mu_g(\lambda, x)
\end{align}
is equivalent to the inequality of \cref{eqn:alg-multiplicity} that has already been established. By placing the nonzero eigenvalues of $H$ in one-to-one correspondence with the elements of $\text{sp}(x) \setminus E$ via \cref{eqn:Superset-Of-Hamiltonian-Eigenvalues} and using the inequality of \cref{eqn:multiplicity}, we find
\begin{align}
   \sum_{\mu \in \text{sp}(H) \setminus \{0\}} \mu_g(\mu, H) 
    &\geq \sum_{\lambda \in \text{sp}(x) \setminus E} \mu_g(\lambda, x).
\end{align}
Suppose either of the preceding inequalities is a strict inequality. Then, inserting this inequality into the right-hand side of \cref{eqn:mult-sum} results in the strict inequality
\begin{align}
    N &> \sum_{\lambda \in E} \mu_g(\lambda, x) + \sum_{\lambda \in \text{sp}(x) \setminus E} \mu_g(\lambda, x) \Leftrightarrow \nonumber \\ 
    N &> \sum_{\lambda \in \text{sp}(x)} \mu_g(\lambda, x),
\end{align}
which would contradict \cref{eqn:contradiction-in-mult}. Thus, the inequalities must be strict equalities, completing the theorem proof.
\end{enumerate}
\end{proof}

%You need (q-1) to not be a zero divisor and q^m != 1 so that you can factor the difference of powers into linear divisors. See e.g. https://proofwiki.org/wiki/Spectral_Mapping_Theorem_for_Polynomials

\bibliography{bib}

\begin{thebibliography}{77}%
\makeatletter
\providecommand \@ifxundefined [1]{%
 \@ifx{#1\undefined}
}%
\providecommand \@ifnum [1]{%
 \ifnum #1\expandafter \@firstoftwo
 \else \expandafter \@secondoftwo
 \fi
}%
\providecommand \@ifx [1]{%
 \ifx #1\expandafter \@firstoftwo
 \else \expandafter \@secondoftwo
 \fi
}%
\providecommand \natexlab [1]{#1}%
\providecommand \enquote  [1]{``#1''}%
\providecommand \bibnamefont  [1]{#1}%
\providecommand \bibfnamefont [1]{#1}%
\providecommand \citenamefont [1]{#1}%
\providecommand \href@noop [0]{\@secondoftwo}%
\providecommand \href [0]{\begingroup \@sanitize@url \@href}%
\providecommand \@href[1]{\@@startlink{#1}\@@href}%
\providecommand \@@href[1]{\endgroup#1\@@endlink}%
\providecommand \@sanitize@url [0]{\catcode `\\12\catcode `\$12\catcode `\&12\catcode `\#12\catcode `\^12\catcode `\_12\catcode `\%12\relax}%
\providecommand \@@startlink[1]{}%
\providecommand \@@endlink[0]{}%
\providecommand \url  [0]{\begingroup\@sanitize@url \@url }%
\providecommand \@url [1]{\endgroup\@href {#1}{\urlprefix }}%
\providecommand \urlprefix  [0]{URL }%
\providecommand \Eprint [0]{\href }%
\providecommand \doibase [0]{https://doi.org/}%
\providecommand \selectlanguage [0]{\@gobble}%
\providecommand \bibinfo  [0]{\@secondoftwo}%
\providecommand \bibfield  [0]{\@secondoftwo}%
\providecommand \translation [1]{[#1]}%
\providecommand \BibitemOpen [0]{}%
\providecommand \bibitemStop [0]{}%
\providecommand \bibitemNoStop [0]{.\EOS\space}%
\providecommand \EOS [0]{\spacefactor3000\relax}%
\providecommand \BibitemShut  [1]{\csname bibitem#1\endcsname}%
\let\auto@bib@innerbib\@empty
%</preamble>
\bibitem [{\citenamefont {Bender}\ and\ \citenamefont {Boettcher}(1998)}]{Bender1998}%
  \BibitemOpen
  \bibfield  {author} {\bibinfo {author} {\bibfnamefont {C.~M.}\ \bibnamefont {Bender}}\ and\ \bibinfo {author} {\bibfnamefont {S.}~\bibnamefont {Boettcher}},\ }\bibfield  {title} {\bibinfo {title} {Real spectra in non-hermitian hamiltonians having $\mathcal{PT}$ symmetry},\ }\href {https://doi.org/10.1103/PhysRevLett.80.5243} {\bibfield  {journal} {\bibinfo  {journal} {Phys. Rev. Lett.}\ }\textbf {\bibinfo {volume} {80}},\ \bibinfo {pages} {5243} (\bibinfo {year} {1998})},\ \Eprint {https://arxiv.org/abs/physics/9712001} {arXiv:physics/9712001 [math-ph]} \BibitemShut {NoStop}%
\bibitem [{\citenamefont {Bender}\ \emph {et~al.}(1999)\citenamefont {Bender}, \citenamefont {Boettcher},\ and\ \citenamefont {Meisinger}}]{Bender1999}%
  \BibitemOpen
  \bibfield  {author} {\bibinfo {author} {\bibfnamefont {C.~M.}\ \bibnamefont {Bender}}, \bibinfo {author} {\bibfnamefont {S.}~\bibnamefont {Boettcher}},\ and\ \bibinfo {author} {\bibfnamefont {P.~N.}\ \bibnamefont {Meisinger}},\ }\bibfield  {title} {\bibinfo {title} {$\mathcal{PT}$-symmetric quantum mechanics},\ }\href {https://doi.org/10.1063/1.532860} {\bibfield  {journal} {\bibinfo  {journal} {J. Math. Phys.}\ }\textbf {\bibinfo {volume} {40}},\ \bibinfo {pages} {2201–2229} (\bibinfo {year} {1999})},\ \Eprint {https://arxiv.org/abs/quant-ph/9809072} {arXiv:quant-ph/9809072 [quant-ph]} \BibitemShut {NoStop}%
\bibitem [{\citenamefont {Garcia}\ and\ \citenamefont {Putinar}(2005)}]{Garcia2005}%
  \BibitemOpen
  \bibfield  {author} {\bibinfo {author} {\bibfnamefont {S.}~\bibnamefont {Garcia}}\ and\ \bibinfo {author} {\bibfnamefont {M.}~\bibnamefont {Putinar}},\ }\bibfield  {title} {\bibinfo {title} {Complex symmetric operators and applications},\ }\href {https://doi.org/10.1090/s0002-9947-05-03742-6} {\bibfield  {journal} {\bibinfo  {journal} {Trans. Am. Math. Soc.}\ }\textbf {\bibinfo {volume} {358}},\ \bibinfo {pages} {1285–1315} (\bibinfo {year} {2005})}\BibitemShut {NoStop}%
\bibitem [{\citenamefont {Kato}(1995)}]{Kato1995}%
  \BibitemOpen
  \bibfield  {author} {\bibinfo {author} {\bibfnamefont {T.}~\bibnamefont {Kato}},\ }\href {https://doi.org/10.1007/978-3-642-66282-9} {\emph {\bibinfo {title} {Perturbation Theory for Linear Operators}}}\ (\bibinfo  {publisher} {Springer Berlin Heidelberg},\ \bibinfo {year} {1995})\BibitemShut {NoStop}%
\bibitem [{\citenamefont {Heiss}(2000)}]{Heiss2000}%
  \BibitemOpen
  \bibfield  {author} {\bibinfo {author} {\bibfnamefont {W.~D.}\ \bibnamefont {Heiss}},\ }\bibfield  {title} {\bibinfo {title} {Repulsion of resonance states and exceptional points},\ }\href {https://doi.org/10.1103/physreve.61.929} {\bibfield  {journal} {\bibinfo  {journal} {Phys. Rev. E}\ }\textbf {\bibinfo {volume} {61}},\ \bibinfo {pages} {929–932} (\bibinfo {year} {2000})},\ \Eprint {https://arxiv.org/abs/quant-ph/9909047} {arXiv:quant-ph/9909047 [quant-ph]} \BibitemShut {NoStop}%
\bibitem [{\citenamefont {Rotter}(2003)}]{Rotter2003}%
  \BibitemOpen
  \bibfield  {author} {\bibinfo {author} {\bibfnamefont {I.}~\bibnamefont {Rotter}},\ }\bibfield  {title} {\bibinfo {title} {Exceptional points and double poles of the $s$ matrix},\ }\href {https://doi.org/10.1103/physreve.67.026204} {\bibfield  {journal} {\bibinfo  {journal} {Phys. Rev. E}\ }\textbf {\bibinfo {volume} {67}},\ \bibinfo {pages} {026204} (\bibinfo {year} {2003})},\ \Eprint {https://arxiv.org/abs/quant-ph/0211197} {arXiv:quant-ph/0211197 [quant-ph]} \BibitemShut {NoStop}%
\bibitem [{\citenamefont {El-Ganainy}\ \emph {et~al.}(2007)\citenamefont {El-Ganainy}, \citenamefont {Makris}, \citenamefont {Christodoulides},\ and\ \citenamefont {Musslimani}}]{ElGanainy2007}%
  \BibitemOpen
  \bibfield  {author} {\bibinfo {author} {\bibfnamefont {R.}~\bibnamefont {El-Ganainy}}, \bibinfo {author} {\bibfnamefont {K.~G.}\ \bibnamefont {Makris}}, \bibinfo {author} {\bibfnamefont {D.~N.}\ \bibnamefont {Christodoulides}},\ and\ \bibinfo {author} {\bibfnamefont {Z.~H.}\ \bibnamefont {Musslimani}},\ }\bibfield  {title} {\bibinfo {title} {Theory of coupled optical pt-symmetric structures},\ }\href {https://doi.org/10.1364/ol.32.002632} {\bibfield  {journal} {\bibinfo  {journal} {Opt. Lett.}\ }\textbf {\bibinfo {volume} {32}},\ \bibinfo {pages} {2632} (\bibinfo {year} {2007})}\BibitemShut {NoStop}%
\bibitem [{\citenamefont {Makris}\ \emph {et~al.}(2008)\citenamefont {Makris}, \citenamefont {El-Ganainy}, \citenamefont {Christodoulides},\ and\ \citenamefont {Musslimani}}]{Makris2008}%
  \BibitemOpen
  \bibfield  {author} {\bibinfo {author} {\bibfnamefont {K.~G.}\ \bibnamefont {Makris}}, \bibinfo {author} {\bibfnamefont {R.}~\bibnamefont {El-Ganainy}}, \bibinfo {author} {\bibfnamefont {D.~N.}\ \bibnamefont {Christodoulides}},\ and\ \bibinfo {author} {\bibfnamefont {Z.~H.}\ \bibnamefont {Musslimani}},\ }\bibfield  {title} {\bibinfo {title} {Beam dynamics in $\mathcal{PT}$-symmetric optical lattices},\ }\href {https://doi.org/10.1103/physrevlett.100.103904} {\bibfield  {journal} {\bibinfo  {journal} {Phys. Rev. Lett.}\ }\textbf {\bibinfo {volume} {100}},\ \bibinfo {pages} {103904} (\bibinfo {year} {2008})}\BibitemShut {NoStop}%
\bibitem [{\citenamefont {R\"{u}ter}\ \emph {et~al.}(2010)\citenamefont {R\"{u}ter}, \citenamefont {Makris}, \citenamefont {El-Ganainy}, \citenamefont {Christodoulides}, \citenamefont {Segev},\ and\ \citenamefont {Kip}}]{Rter2010}%
  \BibitemOpen
  \bibfield  {author} {\bibinfo {author} {\bibfnamefont {C.~E.}\ \bibnamefont {R\"{u}ter}}, \bibinfo {author} {\bibfnamefont {K.~G.}\ \bibnamefont {Makris}}, \bibinfo {author} {\bibfnamefont {R.}~\bibnamefont {El-Ganainy}}, \bibinfo {author} {\bibfnamefont {D.~N.}\ \bibnamefont {Christodoulides}}, \bibinfo {author} {\bibfnamefont {M.}~\bibnamefont {Segev}},\ and\ \bibinfo {author} {\bibfnamefont {D.}~\bibnamefont {Kip}},\ }\bibfield  {title} {\bibinfo {title} {Observation of parity–time symmetry in optics},\ }\href {https://doi.org/10.1038/nphys1515} {\bibfield  {journal} {\bibinfo  {journal} {Nat. Phys.}\ }\textbf {\bibinfo {volume} {6}},\ \bibinfo {pages} {192–195} (\bibinfo {year} {2010})}\BibitemShut {NoStop}%
\bibitem [{\citenamefont {Feng}\ \emph {et~al.}(2017)\citenamefont {Feng}, \citenamefont {El-Ganainy},\ and\ \citenamefont {Ge}}]{Feng2017}%
  \BibitemOpen
  \bibfield  {author} {\bibinfo {author} {\bibfnamefont {L.}~\bibnamefont {Feng}}, \bibinfo {author} {\bibfnamefont {R.}~\bibnamefont {El-Ganainy}},\ and\ \bibinfo {author} {\bibfnamefont {L.}~\bibnamefont {Ge}},\ }\bibfield  {title} {\bibinfo {title} {Non-hermitian photonics based on parity–time symmetry},\ }\href {https://doi.org/10.1038/s41566-017-0031-1} {\bibfield  {journal} {\bibinfo  {journal} {Nat. Photonics}\ }\textbf {\bibinfo {volume} {11}},\ \bibinfo {pages} {752–762} (\bibinfo {year} {2017})}\BibitemShut {NoStop}%
\bibitem [{\citenamefont {El-Ganainy}\ \emph {et~al.}(2018)\citenamefont {El-Ganainy}, \citenamefont {Makris}, \citenamefont {Khajavikhan}, \citenamefont {Musslimani}, \citenamefont {Rotter},\ and\ \citenamefont {Christodoulides}}]{ElGanainy2018}%
  \BibitemOpen
  \bibfield  {author} {\bibinfo {author} {\bibfnamefont {R.}~\bibnamefont {El-Ganainy}}, \bibinfo {author} {\bibfnamefont {K.~G.}\ \bibnamefont {Makris}}, \bibinfo {author} {\bibfnamefont {M.}~\bibnamefont {Khajavikhan}}, \bibinfo {author} {\bibfnamefont {Z.~H.}\ \bibnamefont {Musslimani}}, \bibinfo {author} {\bibfnamefont {S.}~\bibnamefont {Rotter}},\ and\ \bibinfo {author} {\bibfnamefont {D.~N.}\ \bibnamefont {Christodoulides}},\ }\bibfield  {title} {\bibinfo {title} {Non-hermitian physics and pt symmetry},\ }\href {https://doi.org/10.1038/nphys4323} {\bibfield  {journal} {\bibinfo  {journal} {Nat. Phys.}\ }\textbf {\bibinfo {volume} {14}},\ \bibinfo {pages} {11–19} (\bibinfo {year} {2018})}\BibitemShut {NoStop}%
\bibitem [{\citenamefont {\"{O}zdemir}\ \emph {et~al.}(2019)\citenamefont {\"{O}zdemir}, \citenamefont {Rotter}, \citenamefont {Nori},\ and\ \citenamefont {Yang}}]{zdemir2019}%
  \BibitemOpen
  \bibfield  {author} {\bibinfo {author} {\bibfnamefont {{\c{S}}.~K.}\ \bibnamefont {\"{O}zdemir}}, \bibinfo {author} {\bibfnamefont {S.}~\bibnamefont {Rotter}}, \bibinfo {author} {\bibfnamefont {F.}~\bibnamefont {Nori}},\ and\ \bibinfo {author} {\bibfnamefont {L.}~\bibnamefont {Yang}},\ }\bibfield  {title} {\bibinfo {title} {Parity–time symmetry and exceptional points in photonics},\ }\href {https://doi.org/10.1038/s41563-019-0304-9} {\bibfield  {journal} {\bibinfo  {journal} {Nat. Mater.}\ }\textbf {\bibinfo {volume} {18}},\ \bibinfo {pages} {783–798} (\bibinfo {year} {2019})}\BibitemShut {NoStop}%
\bibitem [{\citenamefont {Miri}\ and\ \citenamefont {Alù}(2019)}]{Miri2019}%
  \BibitemOpen
  \bibfield  {author} {\bibinfo {author} {\bibfnamefont {M.-A.}\ \bibnamefont {Miri}}\ and\ \bibinfo {author} {\bibfnamefont {A.}~\bibnamefont {Alù}},\ }\bibfield  {title} {\bibinfo {title} {Exceptional points in optics and photonics},\ }\href {https://doi.org/10.1126/science.aar7709} {\bibfield  {journal} {\bibinfo  {journal} {Science}\ }\textbf {\bibinfo {volume} {363}},\ \bibinfo {pages} {eaar7709} (\bibinfo {year} {2019})}\BibitemShut {NoStop}%
\bibitem [{\citenamefont {Fan}\ \emph {et~al.}(2003)\citenamefont {Fan}, \citenamefont {Suh},\ and\ \citenamefont {Joannopoulos}}]{Fan2003}%
  \BibitemOpen
  \bibfield  {author} {\bibinfo {author} {\bibfnamefont {S.}~\bibnamefont {Fan}}, \bibinfo {author} {\bibfnamefont {W.}~\bibnamefont {Suh}},\ and\ \bibinfo {author} {\bibfnamefont {J.~D.}\ \bibnamefont {Joannopoulos}},\ }\bibfield  {title} {\bibinfo {title} {Temporal coupled-mode theory for the fano resonance in optical resonators},\ }\href {https://doi.org/10.1364/josaa.20.000569} {\bibfield  {journal} {\bibinfo  {journal} {J. Opt. Soc. Am. A}\ }\textbf {\bibinfo {volume} {20}},\ \bibinfo {pages} {569} (\bibinfo {year} {2003})}\BibitemShut {NoStop}%
\bibitem [{\citenamefont {Haus}\ and\ \citenamefont {Huang}(1991)}]{Haus1991}%
  \BibitemOpen
  \bibfield  {author} {\bibinfo {author} {\bibfnamefont {H.}~\bibnamefont {Haus}}\ and\ \bibinfo {author} {\bibfnamefont {W.}~\bibnamefont {Huang}},\ }\bibfield  {title} {\bibinfo {title} {Coupled-mode theory},\ }\href {https://doi.org/10.1109/5.104225} {\bibfield  {journal} {\bibinfo  {journal} {Proc. IEEE}\ }\textbf {\bibinfo {volume} {79}},\ \bibinfo {pages} {1505–1518} (\bibinfo {year} {1991})}\BibitemShut {NoStop}%
\bibitem [{\citenamefont {Wigner}(1932)}]{Wigner1932Zeitumkehr}%
  \BibitemOpen
  \bibfield  {author} {\bibinfo {author} {\bibfnamefont {E.~P.}\ \bibnamefont {Wigner}},\ }\bibfield  {title} {\bibinfo {title} {Über die operation der zeitumkehr in der quantenmechanik},\ }\href@noop {} {\bibfield  {journal} {\bibinfo  {journal} {Nachr. Akad. Wiss. Göttingen Math.-Phys. Kl.}\ ,\ \bibinfo {pages} {546}} (\bibinfo {year} {1932})},\ \bibinfo {note} {{EUDML} \href{https://eudml.org/doc/59401}{59401}}\BibitemShut {NoStop}%
\bibitem [{\citenamefont {Wigner}(1959)}]{Wigner1959}%
  \BibitemOpen
  \bibfield  {author} {\bibinfo {author} {\bibfnamefont {E.~P.}\ \bibnamefont {Wigner}},\ }\href@noop {} {\emph {\bibinfo {title} {Group Theory and Its Application to the Quantum Mechanics of Atomic Spectra}}},\ \bibinfo {edition} {expanded and improved}\ ed.,\ Pure and Applied Physics, Vol. 5\ (\bibinfo  {publisher} {Academic Press Inc.},\ \bibinfo {address} {New York and London},\ \bibinfo {year} {1959})\ pp.\ \bibinfo {pages} {233--236},\ \bibinfo {note} {translated from the German original “Gruppentheorie und ihre Anwendungen auf die Quantenmechanik der Atomspektren”}\BibitemShut {NoStop}%
\bibitem [{\citenamefont {Wigner}(1993)}]{Wigner1993}%
  \BibitemOpen
  \bibfield  {author} {\bibinfo {author} {\bibfnamefont {E.~P.}\ \bibnamefont {Wigner}},\ }\bibinfo {title} {{\"U}ber die operation der zeitumkehr in der quantenmechanik},\ in\ \href {https://doi.org/10.1007/978-3-662-02781-3_15} {\emph {\bibinfo {booktitle} {The Collected Works of Eugene Paul Wigner: Part A: The Scientific Papers}}},\ \bibinfo {editor} {edited by\ \bibinfo {editor} {\bibfnamefont {A.~S.}\ \bibnamefont {Wightman}}}\ (\bibinfo  {publisher} {Springer Berlin Heidelberg},\ \bibinfo {address} {Berlin, Heidelberg},\ \bibinfo {year} {1993})\ pp.\ \bibinfo {pages} {213--226}\BibitemShut {NoStop}%
\bibitem [{\citenamefont {Radjavi}\ and\ \citenamefont {Williams}(1969)}]{Radjavi1969}%
  \BibitemOpen
  \bibfield  {author} {\bibinfo {author} {\bibfnamefont {H.}~\bibnamefont {Radjavi}}\ and\ \bibinfo {author} {\bibfnamefont {J.~P.}\ \bibnamefont {Williams}},\ }\bibfield  {title} {\bibinfo {title} {Products of self-adjoint operators.},\ }\href {https://doi.org/10.1307/mmj/1029000220} {\bibfield  {journal} {\bibinfo  {journal} {Mich. Math. J.}\ }\textbf {\bibinfo {volume} {16}},\ \bibinfo {pages} {177} (\bibinfo {year} {1969})}\BibitemShut {NoStop}%
\bibitem [{\citenamefont {Heisenberg}(1957)}]{Heisenberg1957}%
  \BibitemOpen
  \bibfield  {author} {\bibinfo {author} {\bibfnamefont {W.}~\bibnamefont {Heisenberg}},\ }\bibfield  {title} {\bibinfo {title} {Lee model and quantisation of non linear field equations},\ }\href {https://doi.org/10.1016/0029-5582(87)90060-5} {\bibfield  {journal} {\bibinfo  {journal} {Nucl. Phys.}\ }\textbf {\bibinfo {volume} {4}},\ \bibinfo {pages} {532} (\bibinfo {year} {1957})}\BibitemShut {NoStop}%
\bibitem [{\citenamefont {Scholtz}\ \emph {et~al.}(1992)\citenamefont {Scholtz}, \citenamefont {Geyer},\ and\ \citenamefont {Hahne}}]{Scholtz1992}%
  \BibitemOpen
  \bibfield  {author} {\bibinfo {author} {\bibfnamefont {F.}~\bibnamefont {Scholtz}}, \bibinfo {author} {\bibfnamefont {H.}~\bibnamefont {Geyer}},\ and\ \bibinfo {author} {\bibfnamefont {F.}~\bibnamefont {Hahne}},\ }\bibfield  {title} {\bibinfo {title} {Quasi-hermitian operators in quantum mechanics and the variational principle},\ }\href {https://doi.org/10.1016/0003-4916(92)90284-s} {\bibfield  {journal} {\bibinfo  {journal} {Ann. Phys.}\ }\textbf {\bibinfo {volume} {213}},\ \bibinfo {pages} {74–101} (\bibinfo {year} {1992})}\BibitemShut {NoStop}%
\bibitem [{\citenamefont {Pauli}(1943)}]{Pauli1943DiracFieldQuantization}%
  \BibitemOpen
  \bibfield  {author} {\bibinfo {author} {\bibfnamefont {W.}~\bibnamefont {Pauli}},\ }\bibfield  {title} {\bibinfo {title} {On dirac's new method of field quantization},\ }\href {https://doi.org/10.1103/RevModPhys.15.175} {\bibfield  {journal} {\bibinfo  {journal} {Rev. Mod. Phys.}\ }\textbf {\bibinfo {volume} {15}},\ \bibinfo {pages} {175} (\bibinfo {year} {1943})}\BibitemShut {NoStop}%
\bibitem [{\citenamefont {Mostafazadeh}(2002)}]{Mostafazadeh2002}%
  \BibitemOpen
  \bibfield  {author} {\bibinfo {author} {\bibfnamefont {A.}~\bibnamefont {Mostafazadeh}},\ }\bibfield  {title} {\bibinfo {title} {Pseudo-hermiticity versus pt symmetry: The necessary condition for the reality of the spectrum of a non-hermitian hamiltonian},\ }\href {https://doi.org/10.1063/1.1418246} {\bibfield  {journal} {\bibinfo  {journal} {J. Math. Phys.}\ }\textbf {\bibinfo {volume} {43}},\ \bibinfo {pages} {205–214} (\bibinfo {year} {2002})},\ \Eprint {https://arxiv.org/abs/math-ph/0107001} {arXiv:math-ph/0107001 [math-ph]} \BibitemShut {NoStop}%
\bibitem [{\citenamefont {Mostafazadeh}(2003)}]{Mostafazadeh2003}%
  \BibitemOpen
  \bibfield  {author} {\bibinfo {author} {\bibfnamefont {A.}~\bibnamefont {Mostafazadeh}},\ }\bibfield  {title} {\bibinfo {title} {Pseudo-hermiticity and generalized pt- and cpt-symmetries},\ }\href {https://doi.org/10.1063/1.1539304} {\bibfield  {journal} {\bibinfo  {journal} {J. Math. Phys.}\ }\textbf {\bibinfo {volume} {44}},\ \bibinfo {pages} {974–989} (\bibinfo {year} {2003})},\ \Eprint {https://arxiv.org/abs/math-ph/0209018} {arXiv:math-ph/0209018 [math-ph]} \BibitemShut {NoStop}%
\bibitem [{\citenamefont {Taussky}\ and\ \citenamefont {Parker}(1960)}]{Taussky1960}%
  \BibitemOpen
  \bibfield  {author} {\bibinfo {author} {\bibfnamefont {O.}~\bibnamefont {Taussky}}\ and\ \bibinfo {author} {\bibfnamefont {W.~V.}\ \bibnamefont {Parker}},\ }\bibfield  {title} {\bibinfo {title} {Problem 4846},\ }\href {https://doi.org/10.2307/2308556} {\bibfield  {journal} {\bibinfo  {journal} {Am. Math. Mon.}\ }\textbf {\bibinfo {volume} {67}},\ \bibinfo {pages} {192} (\bibinfo {year} {1960})},\ \bibinfo {note} {{JSTOR} \href{https://www.jstor.org/stable/2308556}{2308556}}\BibitemShut {NoStop}%
\bibitem [{\citenamefont {Drazin}\ and\ \citenamefont {Haynsworth}(1962)}]{Drazin1962}%
  \BibitemOpen
  \bibfield  {author} {\bibinfo {author} {\bibfnamefont {M.~P.}\ \bibnamefont {Drazin}}\ and\ \bibinfo {author} {\bibfnamefont {E.~V.}\ \bibnamefont {Haynsworth}},\ }\bibfield  {title} {\bibinfo {title} {Criteria for the reality of matrix eigenvalues},\ }\href {https://doi.org/10.1007/bf01195188} {\bibfield  {journal} {\bibinfo  {journal} {Math. Z.}\ }\textbf {\bibinfo {volume} {78}},\ \bibinfo {pages} {449} (\bibinfo {year} {1962})},\ \bibinfo {note} {{EUDML} \href{https://eudml.org/doc/170043}{170043}}\BibitemShut {NoStop}%
\bibitem [{\citenamefont {Dieudonn{\'e}}(1961)}]{dieudonne}%
  \BibitemOpen
  \bibfield  {author} {\bibinfo {author} {\bibfnamefont {J.}~\bibnamefont {Dieudonn{\'e}}},\ }\bibfield  {title} {\bibinfo {title} {Quasi-{Hermitian} operators},\ }\href@noop {} {\bibfield  {journal} {\bibinfo  {journal} {Proc. Internat. Sympos. Linear Spaces (Jerusalem, 1960), Pergamon, Oxford}\ }\textbf {\bibinfo {volume} {115122}} (\bibinfo {year} {1961})}\BibitemShut {NoStop}%
\bibitem [{\citenamefont {Koornwinder}(1997)}]{Koornwinder1997}%
  \BibitemOpen
  \bibfield  {author} {\bibinfo {author} {\bibfnamefont {T.~H.}\ \bibnamefont {Koornwinder}},\ }\bibfield  {title} {\bibinfo {title} {Special functions and $q$-commuting variables},\ }in\ \href {https://doi.org/10.1090/fic/014/07} {\emph {\bibinfo {booktitle} {Special Functions, \(q\)-Series and Related Topics}}},\ \bibinfo {series} {Fields Institute Communications}, Vol.~\bibinfo {volume} {14},\ \bibinfo {editor} {edited by\ \bibinfo {editor} {\bibfnamefont {M.~E.~H.}\ \bibnamefont {Ismail}}, \bibinfo {editor} {\bibfnamefont {D.~R.}\ \bibnamefont {Masson}},\ and\ \bibinfo {editor} {\bibfnamefont {M.}~\bibnamefont {Rahman}}}\ (\bibinfo  {publisher} {American Mathematical Society},\ \bibinfo {year} {1997})\ pp.\ \bibinfo {pages} {131--166},\ \Eprint {https://arxiv.org/abs/q-alg/9608008} {arXiv:q-alg/9608008} \BibitemShut {NoStop}%
\bibitem [{\citenamefont {Witten}(1981)}]{Witten1981}%
  \BibitemOpen
  \bibfield  {author} {\bibinfo {author} {\bibfnamefont {E.}~\bibnamefont {Witten}},\ }\bibfield  {title} {\bibinfo {title} {Dynamical breaking of supersymmetry},\ }\href {https://doi.org/10.1016/0550-3213(81)90006-7} {\bibfield  {journal} {\bibinfo  {journal} {Nucl. Phys. B}\ }\textbf {\bibinfo {volume} {188}},\ \bibinfo {pages} {513–554} (\bibinfo {year} {1981})}\BibitemShut {NoStop}%
\bibitem [{\citenamefont {Witten}(1982)}]{Witten1982}%
  \BibitemOpen
  \bibfield  {author} {\bibinfo {author} {\bibfnamefont {E.}~\bibnamefont {Witten}},\ }\bibfield  {title} {\bibinfo {title} {Constraints on supersymmetry breaking},\ }\href {https://doi.org/10.1016/0550-3213(82)90071-2} {\bibfield  {journal} {\bibinfo  {journal} {Nucl. Phys. B}\ }\textbf {\bibinfo {volume} {202}},\ \bibinfo {pages} {253–316} (\bibinfo {year} {1982})}\BibitemShut {NoStop}%
\bibitem [{\citenamefont {Miri}\ \emph {et~al.}(2013)\citenamefont {Miri}, \citenamefont {Heinrich}, \citenamefont {El-Ganainy},\ and\ \citenamefont {Christodoulides}}]{Miri2013}%
  \BibitemOpen
  \bibfield  {author} {\bibinfo {author} {\bibfnamefont {M.-A.}\ \bibnamefont {Miri}}, \bibinfo {author} {\bibfnamefont {M.}~\bibnamefont {Heinrich}}, \bibinfo {author} {\bibfnamefont {R.}~\bibnamefont {El-Ganainy}},\ and\ \bibinfo {author} {\bibfnamefont {D.~N.}\ \bibnamefont {Christodoulides}},\ }\bibfield  {title} {\bibinfo {title} {Supersymmetric optical structures},\ }\href {https://doi.org/10.1103/physrevlett.110.233902} {\bibfield  {journal} {\bibinfo  {journal} {Phys. Rev. Lett.}\ }\textbf {\bibinfo {volume} {110}},\ \bibinfo {pages} {233902} (\bibinfo {year} {2013})},\ \Eprint {https://arxiv.org/abs/1304.6646} {arXiv:1304.6646 [physics.optics]} \BibitemShut {NoStop}%
\bibitem [{\citenamefont {Heinrich}\ \emph {et~al.}(2014)\citenamefont {Heinrich}, \citenamefont {Miri}, \citenamefont {St\"{u}tzer}, \citenamefont {El-Ganainy}, \citenamefont {Nolte}, \citenamefont {Szameit},\ and\ \citenamefont {Christodoulides}}]{Heinrich2014}%
  \BibitemOpen
  \bibfield  {author} {\bibinfo {author} {\bibfnamefont {M.}~\bibnamefont {Heinrich}}, \bibinfo {author} {\bibfnamefont {M.-A.}\ \bibnamefont {Miri}}, \bibinfo {author} {\bibfnamefont {S.}~\bibnamefont {St\"{u}tzer}}, \bibinfo {author} {\bibfnamefont {R.}~\bibnamefont {El-Ganainy}}, \bibinfo {author} {\bibfnamefont {S.}~\bibnamefont {Nolte}}, \bibinfo {author} {\bibfnamefont {A.}~\bibnamefont {Szameit}},\ and\ \bibinfo {author} {\bibfnamefont {D.~N.}\ \bibnamefont {Christodoulides}},\ }\bibfield  {title} {\bibinfo {title} {Supersymmetric mode converters},\ }\href {https://doi.org/10.1038/ncomms4698} {\bibfield  {journal} {\bibinfo  {journal} {Nat. Commun.}\ }\textbf {\bibinfo {volume} {5}},\ \bibinfo {pages} {3698} (\bibinfo {year} {2014})},\ \Eprint {https://arxiv.org/abs/1401.5734} {arXiv:1401.5734 [physics.optics]} \BibitemShut {NoStop}%
\bibitem [{\citenamefont {Datta}\ \emph {et~al.}(2024)\citenamefont {Datta}, \citenamefont {Alizadeh}, \citenamefont {El-Ganainy},\ and\ \citenamefont {Roychowdhury}}]{Datta2024}%
  \BibitemOpen
  \bibfield  {author} {\bibinfo {author} {\bibfnamefont {S.}~\bibnamefont {Datta}}, \bibinfo {author} {\bibfnamefont {M.}~\bibnamefont {Alizadeh}}, \bibinfo {author} {\bibfnamefont {R.}~\bibnamefont {El-Ganainy}},\ and\ \bibinfo {author} {\bibfnamefont {K.}~\bibnamefont {Roychowdhury}},\ }\bibfield  {title} {\bibinfo {title} {A topological route to engineering robust and bright supersymmetric laser arrays},\ }\href {https://doi.org/10.1038/s42005-024-01905-1} {\bibfield  {journal} {\bibinfo  {journal} {Commun. Phys.}\ }\textbf {\bibinfo {volume} {7}},\ \bibinfo {pages} {403} (\bibinfo {year} {2024})},\ \Eprint {https://arxiv.org/abs/2412.12275} {arXiv:2412.12275 [physics.optics]} \BibitemShut {NoStop}%
\bibitem [{\citenamefont {Arkinstall}\ \emph {et~al.}(2017)\citenamefont {Arkinstall}, \citenamefont {Teimourpour}, \citenamefont {Feng}, \citenamefont {El-Ganainy},\ and\ \citenamefont {Schomerus}}]{Arkinstall2017}%
  \BibitemOpen
  \bibfield  {author} {\bibinfo {author} {\bibfnamefont {J.}~\bibnamefont {Arkinstall}}, \bibinfo {author} {\bibfnamefont {M.~H.}\ \bibnamefont {Teimourpour}}, \bibinfo {author} {\bibfnamefont {L.}~\bibnamefont {Feng}}, \bibinfo {author} {\bibfnamefont {R.}~\bibnamefont {El-Ganainy}},\ and\ \bibinfo {author} {\bibfnamefont {H.}~\bibnamefont {Schomerus}},\ }\bibfield  {title} {\bibinfo {title} {Topological tight-binding models from nontrivial square roots},\ }\href {https://doi.org/10.1103/physrevb.95.165109} {\bibfield  {journal} {\bibinfo  {journal} {Phys. Rev. B}\ }\textbf {\bibinfo {volume} {95}},\ \bibinfo {pages} {165109} (\bibinfo {year} {2017})},\ \Eprint {https://arxiv.org/abs/1702.07648} {arXiv:1702.07648 [cond-mat.mes-hall]} \BibitemShut {NoStop}%
\bibitem [{\citenamefont {Zhang}\ \emph {et~al.}(2019)\citenamefont {Zhang}, \citenamefont {Teimourpour}, \citenamefont {Arkinstall}, \citenamefont {Pan}, \citenamefont {Miao}, \citenamefont {Schomerus}, \citenamefont {El‐Ganainy},\ and\ \citenamefont {Feng}}]{Zhang2019}%
  \BibitemOpen
  \bibfield  {author} {\bibinfo {author} {\bibfnamefont {Z.}~\bibnamefont {Zhang}}, \bibinfo {author} {\bibfnamefont {M.~H.}\ \bibnamefont {Teimourpour}}, \bibinfo {author} {\bibfnamefont {J.}~\bibnamefont {Arkinstall}}, \bibinfo {author} {\bibfnamefont {M.}~\bibnamefont {Pan}}, \bibinfo {author} {\bibfnamefont {P.}~\bibnamefont {Miao}}, \bibinfo {author} {\bibfnamefont {H.}~\bibnamefont {Schomerus}}, \bibinfo {author} {\bibfnamefont {R.}~\bibnamefont {El‐Ganainy}},\ and\ \bibinfo {author} {\bibfnamefont {L.}~\bibnamefont {Feng}},\ }\bibfield  {title} {\bibinfo {title} {Experimental realization of multiple topological edge states in a 1d photonic lattice},\ }\href {https://doi.org/10.1002/lpor.201800202} {\bibfield  {journal} {\bibinfo  {journal} {Laser Photonics Rev.}\ }\textbf {\bibinfo {volume} {13}},\ \bibinfo {pages} {1800202} (\bibinfo {year} {2019})},\ \Eprint {https://arxiv.org/abs/1812.05572} {arXiv:1812.05572 [physics.optics]} \BibitemShut {NoStop}%
\bibitem [{\citenamefont {Kremer}\ \emph {et~al.}(2020)\citenamefont {Kremer}, \citenamefont {Petrides}, \citenamefont {Meyer}, \citenamefont {Heinrich}, \citenamefont {Zilberberg},\ and\ \citenamefont {Szameit}}]{Kremer2020}%
  \BibitemOpen
  \bibfield  {author} {\bibinfo {author} {\bibfnamefont {M.}~\bibnamefont {Kremer}}, \bibinfo {author} {\bibfnamefont {I.}~\bibnamefont {Petrides}}, \bibinfo {author} {\bibfnamefont {E.}~\bibnamefont {Meyer}}, \bibinfo {author} {\bibfnamefont {M.}~\bibnamefont {Heinrich}}, \bibinfo {author} {\bibfnamefont {O.}~\bibnamefont {Zilberberg}},\ and\ \bibinfo {author} {\bibfnamefont {A.}~\bibnamefont {Szameit}},\ }\bibfield  {title} {\bibinfo {title} {A square-root topological insulator with non-quantized indices realized with photonic aharonov-bohm cages},\ }\href {https://doi.org/10.1038/s41467-020-14692-4} {\bibfield  {journal} {\bibinfo  {journal} {Nat. Commun.}\ }\textbf {\bibinfo {volume} {11}},\ \bibinfo {pages} {907} (\bibinfo {year} {2020})},\ \Eprint {https://arxiv.org/abs/2006.10731} {arXiv:2006.10731 [cond-mat.mes-hall]} \BibitemShut {NoStop}%
\bibitem [{\citenamefont {Mostafazadeh}\ and\ \citenamefont {\"Oz\c{c}elik}(2006)}]{Mostafazadeh2006PseudoHermitian}%
  \BibitemOpen
  \bibfield  {author} {\bibinfo {author} {\bibfnamefont {A.}~\bibnamefont {Mostafazadeh}}\ and\ \bibinfo {author} {\bibfnamefont {S.}~\bibnamefont {\"Oz\c{c}elik}},\ }\bibfield  {title} {\bibinfo {title} {Explicit realization of pseudo-hermitian and quasi-hermitian quantum mechanics for two-level systems},\ }\href {https://journals.tubitak.gov.tr/physics/vol30/iss5/8} {\bibfield  {journal} {\bibinfo  {journal} {Turk. J. Phys.}\ }\textbf {\bibinfo {volume} {30}},\ \bibinfo {pages} {437} (\bibinfo {year} {2006})},\ \Eprint {https://arxiv.org/abs/quant-ph/0607120} {arXiv:quant-ph/0607120 [quant-ph]} \BibitemShut {NoStop}%
\bibitem [{\citenamefont {Wang}(2013)}]{Wang2013}%
  \BibitemOpen
  \bibfield  {author} {\bibinfo {author} {\bibfnamefont {Q.-H.}\ \bibnamefont {Wang}},\ }\bibfield  {title} {\bibinfo {title} {2 × 2 pt-symmetric matrices and their applications},\ }\href {https://doi.org/10.1098/rsta.2012.0045} {\bibfield  {journal} {\bibinfo  {journal} {Philos. Trans. R. Soc. A}\ }\textbf {\bibinfo {volume} {371}},\ \bibinfo {pages} {20120045} (\bibinfo {year} {2013})}\BibitemShut {NoStop}%
\bibitem [{\citenamefont {Taussky}(1988)}]{Taussky1988}%
  \BibitemOpen
  \bibfield  {author} {\bibinfo {author} {\bibfnamefont {O.}~\bibnamefont {Taussky}},\ }\bibfield  {title} {\bibinfo {title} {How i became a torchbearer for matrix theory},\ }\href {https://doi.org/10.1080/00029890.1988.11972092} {\bibfield  {journal} {\bibinfo  {journal} {Am. Math. Mon.}\ }\textbf {\bibinfo {volume} {95}},\ \bibinfo {pages} {801–812} (\bibinfo {year} {1988})},\ \bibinfo {note} {{JSTOR} \href{https://www.jstor.org/stable/2322895}{2322895}}\BibitemShut {NoStop}%
\bibitem [{\citenamefont {Coulson}\ and\ \citenamefont {Rushbrooke}(1940)}]{Coulson1940}%
  \BibitemOpen
  \bibfield  {author} {\bibinfo {author} {\bibfnamefont {C.~A.}\ \bibnamefont {Coulson}}\ and\ \bibinfo {author} {\bibfnamefont {G.~S.}\ \bibnamefont {Rushbrooke}},\ }\bibfield  {title} {\bibinfo {title} {Note on the method of molecular orbitals},\ }\href {https://doi.org/10.1017/s0305004100017163} {\bibfield  {journal} {\bibinfo  {journal} {Math. Proc. Camb. Philos. Soc.}\ }\textbf {\bibinfo {volume} {36}},\ \bibinfo {pages} {193–200} (\bibinfo {year} {1940})}\BibitemShut {NoStop}%
\bibitem [{\citenamefont {Ruedenberg}\ and\ \citenamefont {Scherr}(1953)}]{Ruedenberg1953}%
  \BibitemOpen
  \bibfield  {author} {\bibinfo {author} {\bibfnamefont {K.}~\bibnamefont {Ruedenberg}}\ and\ \bibinfo {author} {\bibfnamefont {C.~W.}\ \bibnamefont {Scherr}},\ }\bibfield  {title} {\bibinfo {title} {Free-electron network model for conjugated systems. i. theory},\ }\href {https://doi.org/10.1063/1.1699299} {\bibfield  {journal} {\bibinfo  {journal} {J. Chem. Phys.}\ }\textbf {\bibinfo {volume} {21}},\ \bibinfo {pages} {1565–1581} (\bibinfo {year} {1953})}\BibitemShut {NoStop}%
\bibitem [{\citenamefont {Kawabata}\ \emph {et~al.}(2019)\citenamefont {Kawabata}, \citenamefont {Shiozaki}, \citenamefont {Ueda},\ and\ \citenamefont {Sato}}]{Kawabata2019}%
  \BibitemOpen
  \bibfield  {author} {\bibinfo {author} {\bibfnamefont {K.}~\bibnamefont {Kawabata}}, \bibinfo {author} {\bibfnamefont {K.}~\bibnamefont {Shiozaki}}, \bibinfo {author} {\bibfnamefont {M.}~\bibnamefont {Ueda}},\ and\ \bibinfo {author} {\bibfnamefont {M.}~\bibnamefont {Sato}},\ }\bibfield  {title} {\bibinfo {title} {Symmetry and topology in non-hermitian physics},\ }\href {https://doi.org/10.1103/physrevx.9.041015} {\bibfield  {journal} {\bibinfo  {journal} {Phys. Rev. X}\ }\textbf {\bibinfo {volume} {9}},\ \bibinfo {pages} {041015} (\bibinfo {year} {2019})},\ \Eprint {https://arxiv.org/abs/1812.09133} {arXiv:1812.09133 [cond-mat.mes-hall]} \BibitemShut {NoStop}%
\bibitem [{\citenamefont {Znojil}(2002)}]{Znojil2002}%
  \BibitemOpen
  \bibfield  {author} {\bibinfo {author} {\bibfnamefont {M.}~\bibnamefont {Znojil}},\ }\bibfield  {title} {\bibinfo {title} {Non-hermitian supersymmetry and singular, $\mathcal{PT}$-symmetrized oscillators},\ }\href {https://doi.org/10.1088/0305-4470/35/9/320} {\bibfield  {journal} {\bibinfo  {journal} {J. Phys. A}\ }\textbf {\bibinfo {volume} {35}},\ \bibinfo {pages} {2341–2352} (\bibinfo {year} {2002})}\BibitemShut {NoStop}%
\bibitem [{\citenamefont {Znojil}(2004)}]{Znojil2004}%
  \BibitemOpen
  \bibfield  {author} {\bibinfo {author} {\bibfnamefont {M.}~\bibnamefont {Znojil}},\ }\bibfield  {title} {\bibinfo {title} {Script $\mathcal{PT}$-symmetric regularizations in supersymmetric quantum mechanics},\ }\href {https://doi.org/10.1088/0305-4470/37/43/013} {\bibfield  {journal} {\bibinfo  {journal} {J. Phys. A}\ }\textbf {\bibinfo {volume} {37}},\ \bibinfo {pages} {10209–10222} (\bibinfo {year} {2004})}\BibitemShut {NoStop}%
\bibitem [{\citenamefont {Bagarello}(2020)}]{Bagarello2020}%
  \BibitemOpen
  \bibfield  {author} {\bibinfo {author} {\bibfnamefont {F.}~\bibnamefont {Bagarello}},\ }\bibfield  {title} {\bibinfo {title} {Susy for non-hermitian hamiltonians, with a view to coherent states},\ }\href {https://doi.org/10.1007/s11040-020-09353-3} {\bibfield  {journal} {\bibinfo  {journal} {Math. Phys. Anal. Geom.}\ }\textbf {\bibinfo {volume} {23}},\ \bibinfo {pages} {28} (\bibinfo {year} {2020})},\ \Eprint {https://arxiv.org/abs/2007.01677} {arXiv:2007.01677 [math-ph]} \BibitemShut {NoStop}%
\bibitem [{\citenamefont {Sylvester}(1883)}]{Sylvester1883}%
  \BibitemOpen
  \bibfield  {author} {\bibinfo {author} {\bibfnamefont {J.}~\bibnamefont {Sylvester}},\ }\bibfield  {title} {\bibinfo {title} {Xxxix. on the equation to the secular inequalities in the planetary theory},\ }\href {https://doi.org/10.1080/14786448308627430} {\bibfield  {journal} {\bibinfo  {journal} {Philos. Mag.}\ }\textbf {\bibinfo {volume} {16}},\ \bibinfo {pages} {267–269} (\bibinfo {year} {1883})}\BibitemShut {NoStop}%
\bibitem [{\citenamefont {Martínez-Martínez}\ \emph {et~al.}(2024)\citenamefont {Martínez-Martínez}, \citenamefont {Moreno-Rodriguez}, \citenamefont {Méndez-Bermúdez},\ and\ \citenamefont {Benisty}}]{MartnezMartnez2024}%
  \BibitemOpen
  \bibfield  {author} {\bibinfo {author} {\bibfnamefont {C.~T.}\ \bibnamefont {Martínez-Martínez}}, \bibinfo {author} {\bibfnamefont {L.~A.}\ \bibnamefont {Moreno-Rodriguez}}, \bibinfo {author} {\bibfnamefont {J.~A.}\ \bibnamefont {Méndez-Bermúdez}},\ and\ \bibinfo {author} {\bibfnamefont {H.}~\bibnamefont {Benisty}},\ }\bibfield  {title} {\bibinfo {title} {Stability mapping of bipartite tight-binding graphs with losses and gain: $\mathcal{PT}$-symmetry and beyond},\ }\href {https://doi.org/10.1063/5.0199771} {\bibfield  {journal} {\bibinfo  {journal} {Chaos}\ }\textbf {\bibinfo {volume} {34}},\ \bibinfo {pages} {053116} (\bibinfo {year} {2024})},\ \Eprint {https://arxiv.org/abs/2212.13642} {arXiv:2212.13642 [cond-mat.dis-nn]} \BibitemShut {NoStop}%
\bibitem [{\citenamefont {Hasan}\ and\ \citenamefont {Kane}(2010)}]{Hasan2010}%
  \BibitemOpen
  \bibfield  {author} {\bibinfo {author} {\bibfnamefont {M.~Z.}\ \bibnamefont {Hasan}}\ and\ \bibinfo {author} {\bibfnamefont {C.~L.}\ \bibnamefont {Kane}},\ }\bibfield  {title} {\bibinfo {title} {Colloquium: Topological insulators},\ }\href {https://doi.org/10.1103/revmodphys.82.3045} {\bibfield  {journal} {\bibinfo  {journal} {Rev. Mod. Phys.}\ }\textbf {\bibinfo {volume} {82}},\ \bibinfo {pages} {3045–3067} (\bibinfo {year} {2010})},\ \Eprint {https://arxiv.org/abs/1002.3895} {arXiv:1002.3895 [cond-mat.mes-hall]} \BibitemShut {NoStop}%
\bibitem [{\citenamefont {Sutherland}(1986)}]{Sutherland1986}%
  \BibitemOpen
  \bibfield  {author} {\bibinfo {author} {\bibfnamefont {B.}~\bibnamefont {Sutherland}},\ }\bibfield  {title} {\bibinfo {title} {Localization of electronic wave functions due to local topology},\ }\href {https://doi.org/10.1103/physrevb.34.5208} {\bibfield  {journal} {\bibinfo  {journal} {Phys. Rev. B}\ }\textbf {\bibinfo {volume} {34}},\ \bibinfo {pages} {5208–5211} (\bibinfo {year} {1986})}\BibitemShut {NoStop}%
\bibitem [{\citenamefont {Barnett}(2023)}]{Barnett2023PhD}%
  \BibitemOpen
  \bibfield  {author} {\bibinfo {author} {\bibfnamefont {J.~L.}\ \bibnamefont {Barnett}},\ }\emph {\bibinfo {title} {Locality and Exceptional Points in Pseudo-Hermitian Physics}},\ \href {https://doi.org/10.48550/ARXIV.2306.04044} {\bibinfo {type} {Ph.d. thesis}},\ \bibinfo  {school} {University of Waterloo}, \bibinfo {address} {Waterloo, Ontario, Canada} (\bibinfo {year} {2023}),\ \Eprint {https://arxiv.org/abs/2306.04044} {arXiv:2306.04044 [quant-ph]} \BibitemShut {NoStop}%
\bibitem [{\citenamefont {Schomerus}(2013)}]{Schomerus2013}%
  \BibitemOpen
  \bibfield  {author} {\bibinfo {author} {\bibfnamefont {H.}~\bibnamefont {Schomerus}},\ }\bibfield  {title} {\bibinfo {title} {Topologically protected midgap states in complex photonic lattices},\ }\href {https://doi.org/10.1364/ol.38.001912} {\bibfield  {journal} {\bibinfo  {journal} {Opt. Lett.}\ }\textbf {\bibinfo {volume} {38}},\ \bibinfo {pages} {1912} (\bibinfo {year} {2013})},\ \Eprint {https://arxiv.org/abs/1301.0777} {arXiv:1301.0777 [physics.optics]} \BibitemShut {NoStop}%
\bibitem [{\citenamefont {Ge}(2017)}]{Ge2017}%
  \BibitemOpen
  \bibfield  {author} {\bibinfo {author} {\bibfnamefont {L.}~\bibnamefont {Ge}},\ }\bibfield  {title} {\bibinfo {title} {Symmetry-protected zero-mode laser with a tunable spatial profile},\ }\href {https://doi.org/10.1103/physreva.95.023812} {\bibfield  {journal} {\bibinfo  {journal} {Phys. Rev. A}\ }\textbf {\bibinfo {volume} {95}},\ \bibinfo {pages} {023812} (\bibinfo {year} {2017})},\ \Eprint {https://arxiv.org/abs/1610.09717} {arXiv:1610.09717 [physics.optics]} \BibitemShut {NoStop}%
\bibitem [{\citenamefont {Hill}(1969)}]{Hill1969}%
  \BibitemOpen
  \bibfield  {author} {\bibinfo {author} {\bibfnamefont {R.~D.}\ \bibnamefont {Hill}},\ }\bibfield  {title} {\bibinfo {title} {Inertia theory for simultaneously triangulable complex matrices},\ }\href {https://doi.org/10.1016/0024-3795(69)90022-6} {\bibfield  {journal} {\bibinfo  {journal} {Linear Algebra Its Appl.}\ }\textbf {\bibinfo {volume} {2}},\ \bibinfo {pages} {131–142} (\bibinfo {year} {1969})}\BibitemShut {NoStop}%
\bibitem [{\citenamefont {Lieu}(2018)}]{Lieu2018}%
  \BibitemOpen
  \bibfield  {author} {\bibinfo {author} {\bibfnamefont {S.}~\bibnamefont {Lieu}},\ }\bibfield  {title} {\bibinfo {title} {Topological phases in the non-hermitian su-schrieffer-heeger model},\ }\href {https://doi.org/10.1103/physrevb.97.045106} {\bibfield  {journal} {\bibinfo  {journal} {Phys. Rev. B}\ }\textbf {\bibinfo {volume} {97}},\ \bibinfo {pages} {045106} (\bibinfo {year} {2018})},\ \Eprint {https://arxiv.org/abs/1709.03788} {arXiv:1709.03788 [cond-mat.mes-hall]} \BibitemShut {NoStop}%
\bibitem [{\citenamefont {Su}\ \emph {et~al.}(1980)\citenamefont {Su}, \citenamefont {Schrieffer},\ and\ \citenamefont {Heeger}}]{Su1980}%
  \BibitemOpen
  \bibfield  {author} {\bibinfo {author} {\bibfnamefont {W.~P.}\ \bibnamefont {Su}}, \bibinfo {author} {\bibfnamefont {J.~R.}\ \bibnamefont {Schrieffer}},\ and\ \bibinfo {author} {\bibfnamefont {A.~J.}\ \bibnamefont {Heeger}},\ }\bibfield  {title} {\bibinfo {title} {Soliton excitations in polyacetylene},\ }\href {https://doi.org/10.1103/physrevb.22.2099} {\bibfield  {journal} {\bibinfo  {journal} {Phys. Rev. B}\ }\textbf {\bibinfo {volume} {22}},\ \bibinfo {pages} {2099–2111} (\bibinfo {year} {1980})}\BibitemShut {NoStop}%
\bibitem [{\citenamefont {Hatano}\ and\ \citenamefont {Nelson}(1996)}]{HatanoNelson1996Localization}%
  \BibitemOpen
  \bibfield  {author} {\bibinfo {author} {\bibfnamefont {N.}~\bibnamefont {Hatano}}\ and\ \bibinfo {author} {\bibfnamefont {D.~R.}\ \bibnamefont {Nelson}},\ }\bibfield  {title} {\bibinfo {title} {Localization transitions in non‑hermitian quantum mechanics},\ }\href {https://doi.org/10.1103/PhysRevLett.77.570} {\bibfield  {journal} {\bibinfo  {journal} {Phys. Rev. Lett.}\ }\textbf {\bibinfo {volume} {77}},\ \bibinfo {pages} {570} (\bibinfo {year} {1996})},\ \Eprint {https://arxiv.org/abs/cond-mat/9603165} {arXiv:cond-mat/9603165 [cond-mat]} \BibitemShut {NoStop}%
\bibitem [{\citenamefont {R\'ozsa}(1969)}]{Rozsa1969PeriodicContinuants}%
  \BibitemOpen
  \bibfield  {author} {\bibinfo {author} {\bibfnamefont {P.}~\bibnamefont {R\'ozsa}},\ }\bibfield  {title} {\bibinfo {title} {On periodic continuants},\ }\href {https://doi.org/10.1016/0024-3795(69)90030-5} {\bibfield  {journal} {\bibinfo  {journal} {Linear Algebra Its Appl.}\ }\textbf {\bibinfo {volume} {2}},\ \bibinfo {pages} {267} (\bibinfo {year} {1969})}\BibitemShut {NoStop}%
\bibitem [{\citenamefont {Williams}(1969)}]{williams1969operators}%
  \BibitemOpen
  \bibfield  {author} {\bibinfo {author} {\bibfnamefont {J.~P.}\ \bibnamefont {Williams}},\ }\bibfield  {title} {\bibinfo {title} {Operators similar to their adjoints},\ }\href {https://doi.org/10.1090/s0002-9939-1969-0233230-5} {\bibfield  {journal} {\bibinfo  {journal} {Proc. Am. Math. Soc.}\ }\textbf {\bibinfo {volume} {20}},\ \bibinfo {pages} {121} (\bibinfo {year} {1969})},\ \bibinfo {note} {{JSTOR} \href{https://www.jstor.org/stable/2035972}{2035972}}\BibitemShut {NoStop}%
\bibitem [{\citenamefont {Manin}(1987)}]{Manin1987}%
  \BibitemOpen
  \bibfield  {author} {\bibinfo {author} {\bibfnamefont {Y.~I.}\ \bibnamefont {Manin}},\ }\bibfield  {title} {\bibinfo {title} {Some remarks on koszul algebras and quantum groups},\ }\href {https://doi.org/10.5802/aif.1117} {\bibfield  {journal} {\bibinfo  {journal} {Ann. Inst. Fourier}\ }\textbf {\bibinfo {volume} {37}},\ \bibinfo {pages} {191–205} (\bibinfo {year} {1987})}\BibitemShut {NoStop}%
\bibitem [{\citenamefont {Manin}(2018)}]{Manin2018}%
  \BibitemOpen
  \bibfield  {author} {\bibinfo {author} {\bibfnamefont {Y.~I.}\ \bibnamefont {Manin}},\ }\href {https://doi.org/10.1007/978-3-319-97987-8} {\emph {\bibinfo {title} {Quantum Groups and Noncommutative Geometry}}}\ (\bibinfo  {publisher} {Springer International Publishing},\ \bibinfo {year} {2018})\BibitemShut {NoStop}%
\bibitem [{\citenamefont {Gohsrich}\ \emph {et~al.}(2024)\citenamefont {Gohsrich}, \citenamefont {Fauman},\ and\ \citenamefont {Kunst}}]{gohsrich2024exceptional}%
  \BibitemOpen
  \bibfield  {author} {\bibinfo {author} {\bibfnamefont {J.~T.}\ \bibnamefont {Gohsrich}}, \bibinfo {author} {\bibfnamefont {J.}~\bibnamefont {Fauman}},\ and\ \bibinfo {author} {\bibfnamefont {F.~K.}\ \bibnamefont {Kunst}},\ }\href@noop {} {\bibinfo {title} {Exceptional points of any order in a generalized hatano-nelson model}} (\bibinfo {year} {2024}),\ \Eprint {https://arxiv.org/abs/2403.12018} {arXiv:2403.12018 [cond-mat.mes-hall]} \BibitemShut {NoStop}%
\bibitem [{\citenamefont {Sylvester}(1882)}]{Sylvester1882Nonions}%
  \BibitemOpen
  \bibfield  {author} {\bibinfo {author} {\bibfnamefont {J.~J.}\ \bibnamefont {Sylvester}},\ }\bibfield  {title} {\bibinfo {title} {A word on nonions},\ }\href@noop {} {\bibfield  {journal} {\bibinfo  {journal} {Johns Hopkins Univ. Circ.}\ }\textbf {\bibinfo {volume} {1}},\ \bibinfo {pages} {241} (\bibinfo {year} {1882})}\BibitemShut {NoStop}%
\bibitem [{\citenamefont {McIntosh}(1962)}]{MCINTOSH1962169}%
  \BibitemOpen
  \bibfield  {author} {\bibinfo {author} {\bibfnamefont {H.~V.}\ \bibnamefont {McIntosh}},\ }\bibfield  {title} {\bibinfo {title} {On matrices which anticommute with a hamiltonian},\ }\href {https://doi.org/10.1016/0022-2852(62)90019-x} {\bibfield  {journal} {\bibinfo  {journal} {J. Mol. Spectrosc.}\ }\textbf {\bibinfo {volume} {8}},\ \bibinfo {pages} {169} (\bibinfo {year} {1962})}\BibitemShut {NoStop}%
\bibitem [{\citenamefont {Sopena}(2001)}]{Sopena2001}%
  \BibitemOpen
  \bibfield  {author} {\bibinfo {author} {\bibfnamefont {E.}~\bibnamefont {Sopena}},\ }\bibfield  {title} {\bibinfo {title} {Oriented graph coloring},\ }\href {https://doi.org/10.1016/s0012-365x(00)00216-8} {\bibfield  {journal} {\bibinfo  {journal} {Discrete Math.}\ }\textbf {\bibinfo {volume} {229}},\ \bibinfo {pages} {359–369} (\bibinfo {year} {2001})}\BibitemShut {NoStop}%
\bibitem [{\citenamefont {Kac}\ and\ \citenamefont {Cheung}(2002)}]{kac2002quantum}%
  \BibitemOpen
  \bibfield  {author} {\bibinfo {author} {\bibfnamefont {V.~G.}\ \bibnamefont {Kac}}\ and\ \bibinfo {author} {\bibfnamefont {P.}~\bibnamefont {Cheung}},\ }\href {https://doi.org/10.1007/978-1-4613-0071-7} {\emph {\bibinfo {title} {Quantum calculus}}},\ Vol.\ \bibinfo {volume} {113}\ (\bibinfo  {publisher} {Springer},\ \bibinfo {year} {2002})\BibitemShut {NoStop}%
\bibitem [{\citenamefont {Schutzenberger}(1953)}]{schutzenberger1953interpretation}%
  \BibitemOpen
  \bibfield  {author} {\bibinfo {author} {\bibfnamefont {M.-P.}\ \bibnamefont {Schutzenberger}},\ }\bibfield  {title} {\bibinfo {title} {Une interpretation de certaines solutions de lequation fonctionnelle-f (x+ y)= f (x) f (y)},\ }\href@noop {} {\bibfield  {journal} {\bibinfo  {journal} {C. R. Acad. Sci.}\ }\textbf {\bibinfo {volume} {236}},\ \bibinfo {pages} {352} (\bibinfo {year} {1953})}\BibitemShut {NoStop}%
\bibitem [{\citenamefont {Manin}(1992)}]{manin1991notes}%
  \BibitemOpen
  \bibfield  {author} {\bibinfo {author} {\bibfnamefont {Y.~I.}\ \bibnamefont {Manin}},\ }\bibfield  {title} {\bibinfo {title} {Notes on quantum groups and quantum de {Rham} complexes},\ }\href@noop {} {\bibfield  {journal} {\bibinfo  {journal} {Teoret. Mat. Fiz.}\ }\textbf {\bibinfo {volume} {92}},\ \bibinfo {pages} {425} (\bibinfo {year} {1992})}\BibitemShut {NoStop}%
\bibitem [{\citenamefont {Gauß}(1808)}]{Gauß1808}%
  \BibitemOpen
  \bibfield  {author} {\bibinfo {author} {\bibfnamefont {C.~F.}\ \bibnamefont {Gauß}},\ }\href@noop {} {\emph {\bibinfo {title} {Summatio quarumdam serierum singularium}}}\ (\bibinfo  {publisher} {Dieterich},\ \bibinfo {year} {1808})\ \bibinfo {note} {{EUDML} \href{https://eudml.org/doc/203313}{203313}}\BibitemShut {NoStop}%
\bibitem [{\citenamefont {Araujo-Regado}\ \emph {et~al.}(2023)\citenamefont {Araujo-Regado}, \citenamefont {Khan},\ and\ \citenamefont {Wall}}]{AraujoRegado2023}%
  \BibitemOpen
  \bibfield  {author} {\bibinfo {author} {\bibfnamefont {G.}~\bibnamefont {Araujo-Regado}}, \bibinfo {author} {\bibfnamefont {R.}~\bibnamefont {Khan}},\ and\ \bibinfo {author} {\bibfnamefont {A.~C.}\ \bibnamefont {Wall}},\ }\bibfield  {title} {\bibinfo {title} {Cauchy slice holography: a new ads/cft dictionary},\ }\href {https://doi.org/10.1007/jhep03(2023)026} {\bibfield  {journal} {\bibinfo  {journal} {J. High Energy Phys.}\ }\textbf {\bibinfo {volume} {2023}}\bibfield  {number} {\bibinfo  {number} { (3)}},\ }\Eprint {https://arxiv.org/abs/2204.00591} {arXiv:2204.00591 [quant-ph]} \BibitemShut {NoStop}%
\bibitem [{\citenamefont {Pauli}(1927)}]{Pauli1927}%
  \BibitemOpen
  \bibfield  {author} {\bibinfo {author} {\bibfnamefont {W.}~\bibnamefont {Pauli}},\ }\bibfield  {title} {\bibinfo {title} {Zur quantenmechanik des magnetischen elektrons},\ }\href {https://doi.org/10.1007/bf01397326} {\bibfield  {journal} {\bibinfo  {journal} {Z. Phys.}\ }\textbf {\bibinfo {volume} {43}},\ \bibinfo {pages} {601–623} (\bibinfo {year} {1927})}\BibitemShut {NoStop}%
\bibitem [{\citenamefont {Bender}\ and\ \citenamefont {Mannheim}(2010)}]{Bender2010}%
  \BibitemOpen
  \bibfield  {author} {\bibinfo {author} {\bibfnamefont {C.~M.}\ \bibnamefont {Bender}}\ and\ \bibinfo {author} {\bibfnamefont {P.~D.}\ \bibnamefont {Mannheim}},\ }\bibfield  {title} {\bibinfo {title} {$\mathcal{PT}$-symmetry and necessary and sufficient conditions for the reality of energy eigenvalues},\ }\href {https://doi.org/10.1016/j.physleta.2010.02.032} {\bibfield  {journal} {\bibinfo  {journal} {Phys. Lett. A}\ }\textbf {\bibinfo {volume} {374}},\ \bibinfo {pages} {1616–1620} (\bibinfo {year} {2010})},\ \Eprint {https://arxiv.org/abs/0902.1365} {arXiv:0902.1365 [hep-th]} \BibitemShut {NoStop}%
\bibitem [{\citenamefont {Pop}\ and\ \citenamefont {Furdui}(2016)}]{PopFurdui_SquareMatrices2}%
  \BibitemOpen
  \bibfield  {author} {\bibinfo {author} {\bibfnamefont {V.}~\bibnamefont {Pop}}\ and\ \bibinfo {author} {\bibfnamefont {O.}~\bibnamefont {Furdui}},\ }\href {https://doi.org/10.1007/978-3-319-54939-2} {\emph {\bibinfo {title} {Square Matrices of Order 2: Theory, Applications and Problems}}}\ (\bibinfo  {publisher} {Springer},\ \bibinfo {address} {Cham, Switzerland},\ \bibinfo {year} {2016})\BibitemShut {NoStop}%
\bibitem [{\citenamefont {Doran}(2018)}]{Doran2018}%
  \BibitemOpen
  \bibfield  {author} {\bibinfo {author} {\bibfnamefont {R.}~\bibnamefont {Doran}},\ }\href {https://doi.org/10.1201/9781315139043} {\emph {\bibinfo {title} {Characterizations of $\text{C}^*$-Algebras}}}\ (\bibinfo  {publisher} {CRC Press},\ \bibinfo {year} {2018})\BibitemShut {NoStop}%
\bibitem [{\citenamefont {Bru}\ and\ \citenamefont {Alberto~de Siqueira~Pedra}(2023)}]{Bru2023}%
  \BibitemOpen
  \bibfield  {author} {\bibinfo {author} {\bibfnamefont {J.-B.}\ \bibnamefont {Bru}}\ and\ \bibinfo {author} {\bibfnamefont {W.}~\bibnamefont {Alberto~de Siqueira~Pedra}},\ }\href {https://doi.org/10.1007/978-3-031-28949-1} {\emph {\bibinfo {title} {$C^*$-Algebras and Mathematical Foundations of Quantum Statistical Mechanics: An Introduction}}}\ (\bibinfo  {publisher} {Springer International Publishing},\ \bibinfo {year} {2023})\BibitemShut {NoStop}%
\bibitem [{\citenamefont {Karuvade}\ \emph {et~al.}(2022)\citenamefont {Karuvade}, \citenamefont {Alase},\ and\ \citenamefont {Sanders}}]{Karuvade2022}%
  \BibitemOpen
  \bibfield  {author} {\bibinfo {author} {\bibfnamefont {S.}~\bibnamefont {Karuvade}}, \bibinfo {author} {\bibfnamefont {A.}~\bibnamefont {Alase}},\ and\ \bibinfo {author} {\bibfnamefont {B.~C.}\ \bibnamefont {Sanders}},\ }\bibfield  {title} {\bibinfo {title} {Observing a changing hilbert-space inner product},\ }\href {https://doi.org/10.1103/physrevresearch.4.013016} {\bibfield  {journal} {\bibinfo  {journal} {Phys. Rev. Res.}\ }\textbf {\bibinfo {volume} {4}},\ \bibinfo {pages} {013016} (\bibinfo {year} {2022})},\ \Eprint {https://arxiv.org/abs/2101.00015} {arXiv:2101.00015 [quant-ph]} \BibitemShut {NoStop}%
\bibitem [{\citenamefont {Bian}\ \emph {et~al.}(2020)\citenamefont {Bian}, \citenamefont {Xiao}, \citenamefont {Wang}, \citenamefont {Zhan}, \citenamefont {Onanga}, \citenamefont {Ruzicka}, \citenamefont {Yi}, \citenamefont {Joglekar},\ and\ \citenamefont {Xue}}]{bian2019time}%
  \BibitemOpen
  \bibfield  {author} {\bibinfo {author} {\bibfnamefont {Z.}~\bibnamefont {Bian}}, \bibinfo {author} {\bibfnamefont {L.}~\bibnamefont {Xiao}}, \bibinfo {author} {\bibfnamefont {K.}~\bibnamefont {Wang}}, \bibinfo {author} {\bibfnamefont {X.}~\bibnamefont {Zhan}}, \bibinfo {author} {\bibfnamefont {F.~A.}\ \bibnamefont {Onanga}}, \bibinfo {author} {\bibfnamefont {F.}~\bibnamefont {Ruzicka}}, \bibinfo {author} {\bibfnamefont {W.}~\bibnamefont {Yi}}, \bibinfo {author} {\bibfnamefont {Y.~N.}\ \bibnamefont {Joglekar}},\ and\ \bibinfo {author} {\bibfnamefont {P.}~\bibnamefont {Xue}},\ }\bibfield  {title} {\bibinfo {title} {Conserved quantities in parity-time symmetric systems},\ }\href {https://doi.org/10.1103/physrevresearch.2.022039} {\bibfield  {journal} {\bibinfo  {journal} {Phys. Rev. Res.}\ }\textbf {\bibinfo {volume} {2}},\ \bibinfo {pages} {022039} (\bibinfo {year} {2020})},\ \Eprint {https://arxiv.org/abs/1903.09806} {arXiv:1903.09806 [quant-ph]} \BibitemShut {NoStop}%
\bibitem [{\citenamefont {De~Nittis}\ and\ \citenamefont {Polo~Ojito}(2023)}]{DeNittis2023}%
  \BibitemOpen
  \bibfield  {author} {\bibinfo {author} {\bibfnamefont {G.}~\bibnamefont {De~Nittis}}\ and\ \bibinfo {author} {\bibfnamefont {D.}~\bibnamefont {Polo~Ojito}},\ }\bibfield  {title} {\bibinfo {title} {About the notion of eigenstates for $c^*$-algebras and some application in quantum mechanics},\ }\href {https://doi.org/10.1063/5.0153219} {\bibfield  {journal} {\bibinfo  {journal} {J. Math. Phys.}\ }\textbf {\bibinfo {volume} {64}},\ \bibinfo {pages} {083506} (\bibinfo {year} {2023})},\ \Eprint {https://arxiv.org/abs/2304.02685} {arXiv:2304.02685 [math.CO]} \BibitemShut {NoStop}%
\end{thebibliography}%

\end{document}